\documentclass[11pt,a4paper]{article}

\usepackage{amsfonts,amsmath,amsthm}

\usepackage{pstricks,pst-node,pst-text,pst-3d}

\theoremstyle{plain}
\newtheorem{theorem}{Theorem}[section]
\newtheorem{lemma}[theorem]{Lemma}
\newtheorem{proposition}[theorem]{Proposition}
\newtheorem{corollary}[theorem]{Corollary}

\theoremstyle{definition}
\newtheorem{example}[theorem]{Example}
\newtheorem{remark}[theorem]{Remark}

\def\calA{{\mathcal A}}
\def\calB{{\mathcal B}}
\def\calC{{\mathcal C}}

\def\calF{{\mathcal F}}

\def\calK{{\mathcal K}}
\def\calL{{\mathcal L}}

\def\calO{{\mathcal O}}

\def\calT{{\mathcal T}}

\def\calV{{\mathcal V}}
\def\calW{{\mathcal W}}

\def\bfA{{\mathbf A}}

\def\bfF{{\mathbf F}}
\def\bfG{{\mathbf G}}
\def\bfH{{\mathbf H}}

\def\bfN{{\mathbf N}}

\def\bfa{{\mathbf a}}

\def\bf1{{\mathbf 1}}
\def\bf2{{\mathbf 2}}

\def\rmr{{\mathrm r}}

\def\frA{{\mathfrak A}}
\def\frB{{\mathfrak B}}
\def\frC{{\mathfrak C}}

\def\ra{\rightarrow}
\def\la{\leftarrow}

\def\LRa{\Leftrightarrow}

\def\LLRa{\Longleftrightarrow}

\def\es{\emptyset}
\def\se{\subseteq}

\def\ve{\varepsilon}

\def\Si{\Sigma}
\def\si{\sigma}
\def\Om{\Omega}
\def\om{\omega}
\def\Ga{\Gamma}

\def\rSi{\mathrm{r}(\Si)}

\def\dt{\mathrm{dt}}

\def\hg{\mathrm{hg}}
\def\leaf{\mathrm{leaf}}

\def\lr{\mathrm{lr}}
\def\pr{\mathrm{pr}}
\def\ran{\mathrm{ran}}

\def\root{\mathrm{root}}
\def\run{\mathrm{run}}

\def\sub{\mathrm{sub}}

\def\supp{\mathrm{supp}}
\def\tde{\widetilde{\delta}}
\def\tdei{\tilde{\delta}^{-1}}
\def\tDe{\widetilde{\Delta}}

\def\wh{\widehat}

\def\SX{\Sigma X}

\def\OY{\Omega Y}
\def\GX{\Gamma X}

\def\SXt{T_\Sigma(X)}

\def\OYt{T_\Omega(Y)}
\def\GXt{T_\Gamma(X)}

\def\SXc{C_\Sigma(X)}

\def\algA{\mathcal{A} = (A,\Sigma)}
\def\algB{\mathcal{B} = (B,\Sigma)}
\def\nalgA{\mathfrak{A} = (A,\Sigma)}
\def\nalgB{\mathfrak{B} = (B,\Sigma)}

\def\pA{\wp\frA = (\wp(A),\Si)}

\def\SXta{\calT_\Sigma(X)}

\def\recG{\mathbf{G} = (\calB,b_0,\pi)}

\def\recA{\mathbf{A} = (\calA,a_0,\alpha)}
\def\recNA{\mathbf{NA} = (\frA,I,\alpha)}

\def\recF{\mathbf{F} = (\calA,a_0,\om)}
\def\recNF{\mathbf{NF} = (\frA,I,\om)}

\def\recNG{\mathbf{NG} = (A,\Si,X,\gamma,\iota,\om)}

\def\recpNF{\wp\mathbf{NF} = (\wp\frA,I,\pi)}
\def\pNF{\wp\NF}

\def\NA{\mathbf{NA}}
\def\NF{\mathbf{NF}}
\def\NG{\mathbf{NG}}
\def\NH{\mathbf{NH}}

\begin{document}

\title{Fuzzy Deterministic Top-down Tree Automata}

\author{Eija Jurvanen\footnote{jurvanen@utu.fi}\:  and Magnus Steinby\footnote{steinby@utu.fi}\\
\emph{Department of Mathematics and Statistics}\\
\emph{University of Turku}\\
\emph{20014 Turku, Finland}}

\date{}

\maketitle

\begin{abstract}
In this paper we introduce and study fuzzy deterministic top-down (DT) tree automata over a lattice $\calL$. The $\calL$-fuzzy tree languages recognized by these automata are said to be DT-recognizable, and they form a proper subfamily $DRec_\calL$ of the family of $Rec_\calL$ of all regular $\calL$-fuzzy tree languages. We prove a Pumping Lemma for $DRec_\calL$  from which several decidability results follow. The closure properties of $DRec_\calL$ under various operations are established. We also characterize DT-recognizability in terms of $\calL$-fuzzy path languages, and prove that the path closure of any regular $\calL$-fuzzy tree language is DT-recognizable, and that it is decidable whether a regular $\calL$-fuzzy tree language is DT-recognizable. In most of the paper, $\calL$ is just any nontrivial bounded lattice, but sometimes it is assumed to be distributive or even a bounded chain.
\end{abstract}

\smallskip


\section{Introduction}\label{Introduction}
A finite tree automaton may process input trees either bottom-up (frontier-to-root) starting at the leaves and proceeding towards the root, or top-down (root-to-frontier) starting at the root and moving towards the leaves.  While the family \emph{Rec} of the recognizable (or regular) tree languages is defined both by deterministic and nondeterministic bottom-up automata as well as by nondeterministic top-down automata, the deterministic top-down (DT) automata give a proper subfamily \emph{DRec} of \emph{Rec}. These facts were established for binary trees by Magidor and Moran \cite{MaMo69} who also observed that any member $T$ of \emph{DRec} is fully determined by the labeled paths from the root to a leaf appearing in its trees: if every path in a given tree $t$ appears in at least one tree belonging to $T$, then also $t$ must be in $T$. This characterization of \emph{DRec}-languages was explored further in a general form in \cite{Cou78} and \cite{Vir80}. The path closure requirement means that even some very simple regular tree languages cannot be recognized by a DT automaton. On the other hand, the derivations in any context-free grammar can be represented by a DT-recognizable tree language \cite{GeSt84}, the regular types for logic programs considered in \cite{YaSh91} are actually DT-recognizable tree languages, and in \cite{MaNS08} it is argued that certain DT automata suffice to express markup language schema definitions. Moreover, DT tree automata and the family \emph{DRec} have some interesting properties of their own, and there is a fairly extensive literature on DT automata and their extensions. We refer the reader to items \cite{Cou78,Esi86,GeSt78,GeSt84,GeSt97,Jur92,Jur95,MaMo69,MaNS08,Ste17,Vir80} of the bibliography and the references in them.

Various types of fuzzy automata and languages have been studied already since the 1960s, the first works appearing quite soon after Zadeh \cite{Zad65} had introduced the idea of fuzziness. For surveys and bibliographies covering the topic, one may consult \cite{Asv96}, \cite{MoMa02} and \cite{Rah09}. In particular, fuzzy tree recognizers have also been studied by some authors (cf. \cite{BoBo10,EsLi07,InFu75,Rah09}, for example).
In this paper we introduce and study finite fuzzy deterministic top-down tree recognizers and the fuzzy tree languages defined by them.

Much of the theory of fuzzy automata and languages has been done within Zadeh's classical setting where the degrees of membership are real numbers taken from the interval  $[0,1]$. However, also more general notions of fuzziness are common. Here we shall consider fuzzy tree languages with membership degrees in a bounded lattice $\calL = (L,\leq)$. Thus they are $\calL$-fuzzy sets in the sense of Goguen \cite{Gog67}. It is customary to assume that the lattice $\calL$ is distributive and complete  \cite{Gog67,Rah09}, or even completely distributive \cite{EsLi07}, but usually we need neither distributivity nor completeness because our automata compute the degrees of acceptance  using the meet-operation only. This means also that there is no need to assume that $\calL$ is locally finite, i.e., that finitely generated sublattices of $\calL$ are finite. Therefore, in most of the paper,  $\calL$ is just any nontrivial lattice with a least element $0$ and a greatest element $1$. However, when dealing with nondeterministic fuzzy tree recognizers and general regular fuzzy tree languages, we assume that $\calL$ is distributive,  and in Section \ref{se:Path closure} it is a bounded chain.

In what follows, DT stands for \emph{deterministic top-down} and NDT for \emph{nondeterministic top-down}.

In Section \ref{se:TreesContextsPaths} we define several concepts and symbols pertaining to trees as well as some general notation.

All our tree recognizers, both classical and fuzzy, are essentially deterministic or nondeterministic top-down algebras equipped with some starting and acceptance mechanisms. Such algebras are considered in Section \ref{se:Top-down algebras and recognizers}. For DT algebras we define the notion of a run tree, and with any NDT algebra we associate a subset algebra which is a DT algebra.
Then we present the definitions of NDT and DT tree recognizers. The former define the family $Rec$ of all recognizable (regular) tree languages while the latter yields the subfamily $DRec$ of DT-recognizable tree languages.

For any ranked alphabet $\Si$ and any leaf alphabet $X$, the set of $\SX$-trees is denoted by $\SXt$, and  an
$\calL$-fuzzy $\SX$-tree language is any mapping $\Phi : \SXt \ra L$. The $\calL$-fuzzy NDT tree recognizers defined in Section \ref{se:L-fuzzy tree languages} are  equivalent to the $\calL$-fuzzy bottom-up tree recognizers of \cite{EsLi07}, and hence they recognize the recognizable, or regular, $\calL$-fuzzy tree languages. We also show that if $\calL$ is distributive, then these recognizers can be simplified by making the state transitions and the set of initial states crisp. We conclude the section by showing that the equivalence problem of $\calL$-fuzzy NDT tree recognizers is decidable.

In Section \ref{se:L-fuzzy DT tree recognizers} we define our DT tree recognizers and the family $DRec_\calL$ of DT-recognizable $\calL$-fuzzy tree languages, and establish some of their basic properties. In particular, we prove a Pumping Lemma from which several decidability results are derived,  including the decidability of the Emptiness, Finiteness and Equivalence Problems of DT tree recognizers.

In Section \ref{se:DRec_L, Rec_L and DRec} we study connections between $\calL$-fuzzy DT-recognizable, $\calL$-fuzzy recognizable tree languages,  and usual DT-recognizable tree languages.
In Section \ref{se:Closure properties} we define some operations on $\calL$-fuzzy tree languages. The operations are natural generalizations of well-known tree language operations and similar to the corresponding operations on tree series. For example top-concatenations, translations and inverse translations, and inverse tree homomorphisms are shown to preserve DT-recognizability. On the other hand, as one would expect, all the negative closure properties of $DRec$ are inherited by $DRec_\calL$.

The ordinary DT-recognizable tree languages are precisely the path closed regular tree languages; a tree language $T$ is path closed if it contains any tree $t$ such that every labeled path in $t$ leading from the root to a leaf appears in some member of $T$.  In Section \ref{se:Fuzzy path languages} we introduce $\calL$-fuzzy path languages and path closures. We also associate with any $\calL$-fuzzy DT tree recognizer $\bfF$ an $\calL$-fuzzy path language $\Lambda_\bfF$, and show that this is recognized by a unary $\calL$-fuzzy DT recognizer obtained from $\bfF$. The main result of this section is a characterization of the DT-recognizable $\calL$-fuzzy tree languages in terms of $\calL$-fuzzy path languages from which it follows that any DT-recognizable $\calL$-fuzzy tree language is path closed.

In Section \ref{se:Path closure} we study further fuzzy path languages and DT-recognizability, now under the assumption that the lattice of membership values is a bounded chain $\calC = (C,\leq)$. For any $\calC$-fuzzy NDT tree recognizer $\NF$, we introduce the $\calC$-fuzzy path language $\Lambda_{\NF}$ defined by $\NF$. We also define the subset recognizer $\wp\NF$ of $\NF$, which is a $\calC$-fuzzy DT tree recognizer, and show that $\Lambda_{\wp\NF} = \Lambda_{\NF}$.

In the theory of crisp DT tree recognizers many results depend on the `normalization' of NDT and DT tree recognizers (cf. \cite{GeSt78,GeSt84,Jur95}). When the membership value lattice is a chain $\calC$, we may introduce a fuzzy version of this notion, and show that any $\calC$-fuzzy NDT tree recognizer $\NF$ can be normalized. Then we prove that if $\NF$ is normalized, then $\wp\NF$ recognizes the $\calC$-fuzzy path closure of the $\calC$-fuzzy tree language recognized by $\NF$. From this result we can infer that the $\calC$-fuzzy path closure of any regular $\calC$-fuzzy tree language is DT-recognizable, and that a regular $\calC$-fuzzy tree language is DT-recognizable exactly in case it is path closed, and that it is decidable whether a recognizable $\calC$-fuzzy tree language is DT-recognizable.

We also consider normalized $\calC$-fuzzy DT tree recognizers. Any $\calC$-fuzzy DT tree recognizer can be normalized, and we show that normalized $\calC$-fuzzy DT tree recognizers have some special properties.

\section{Trees, contexts and paths}\label{se:TreesContextsPaths}

Let us first introduce some general notation.
We may write $A := B$ to emphasize that $A$ is defined to be equal to $B$. The cardinality of a set $A$ is denoted by $|A|$, and the set of all subsets of $A$ by $\wp(A)$. For any nonnegative integer $n$, let $[n] := \{1,\ldots,n\}$.

For any direct product $A_1\times \cdots \times A_m$ ($m\geq 1$) and any $i\in [m]$, the $i^{th}$ \emph{projection map} $A_1\times \cdots \times A_m \ra  A_i, (a_1,\ldots,a_m) \mapsto a_i,$ is denoted by $\pr_i$, and it is also extended to sets of vectors:
\[
\pr_i : \wp(A_1\times \cdots \times A_m) \ra  \wp(A_i), V \mapsto \{\pr_i(\bfa) \mid \bfa \in V\}.
\]

For any alphabet $X$, $X^*$ denotes the set of all (finite) words over $X$. The empty word is denoted by $\ve$. Subsets of $X^*$ are called \emph{languages} (over $X$).

Let $\Sigma$ be a {\em ranked alphabet}, i.e., a finite set of operation symbols, which does not contain nullary symbols. For each $m \geq 1$, $\Sigma_m$ denotes the set of $m$-ary symbols in $\Sigma$. If $\Si$ consists of the symbols $f_1,\ldots,f_k$ of the respective arities $m_1,\ldots,m_k$, we may write $\Si = \{f_1/m_1,\ldots,f_k/m_k\}$. The set $\mathrm{r}(\Si) := \{m \in \rmr(\Si) \mid \Si_m \neq \es\}$ is called the \emph{rank type} of $\Si$. If $\rmr(\Si) = \{1\}$, then $\Si$ is said to be \emph{unary}.
In what follows, $\Sigma$ is always a ranked alphabet and $X$ is an ordinary finite non-empty alphabet, called a {\em leaf alphabet}, disjoint from $\Sigma$. The set $T_\Sigma(X)$ of $\SX${\em -trees}  is the least set such that $X\subseteq T_\Sigma(X)$, and   $f(t_1,\ldots,t_m)\in T_\Sigma(X)$ for all $m \in \rmr(\Si)$, $f\in \Sigma_m$ and $t_1,\ldots,t_m\in T_\Sigma(X)$.
A $\Sigma X${\em -tree language} is any subset of $T_\Sigma(X)$. We also speak about \emph{trees} and \emph{tree languages} without specifying the alphabets. A \emph{family of tree languages} is a map $\calV$ that assigns to each pair $\Si,X$ a set $\calV(\Si,X)$ of $\SX$-tree languages. We write $\calV = \{\calV(\Si,X)\}$. For any families of tree languages $\calV$ and $\calW$, $\calV \se \calW$ means that $\calV(\Si,X) \se \calW(\Si,X)$ for all $\Si$ and $X$, and unions and intersections of such families are defined by similar componentwise conditions.

The  \emph{root (symbol)} $\root(t) (\in \Si \cup X)$, the set of \emph{subtrees} $\sub(t)$, the set of \emph{leaf symbols} $\leaf(t) (\se X)$, and the \emph{height} $\hg(t)$ of a $\SX$-tree $t$ are defined as follows: for any $x\in X$ and $t = f(t_1,\ldots,t_m)$,
\begin{itemize}
  \item[(1)] $\root(x) = x$, $\sub(x) = \leaf(x) = \{x\}$, and $\hg(x) = 0$,  and
  \item[(2)] $\root(t) = f$, $\sub(t) = \{t\}\cup \sub(t_1) \cup \ldots \cup \sub(t_m)$, $\leaf(t) = \leaf(t_1)\cup\ldots\cup\leaf(t_m)$, and $\hg(t) = \max\{\hg(t_1),\ldots,\hg(t_m)\}+1$.
\end{itemize}

Let $\xi$ be a symbol that does not belong to our  alphabets $\Si$ or $X$. A \emph{$\SX$-context} is a $\Si(X\cup\{\xi\}$)-tree in which $\xi$ appears exactly once. The set of $\SX$-contexts is denoted by $\SXc$. For any $p,q\in \SXc$ and $t\in \SXt$, let $p(t)$ and $p(q)$ be the $\SX$-tree and the $\SX$-context obtained from $p$ by replacing the $\xi$ by $t$ and $q$, respectively. Furthermore, let $p\cdot q := p(q)$ and $p\cdot t := p(t)$. The \emph{depth} $\dt(p)$ of a $\SX$-context is $0$ if $p = \xi$ and $\dt(q)+1$ if $p = f(\ldots,\xi,\ldots)\cdot q$ with $f\in \Si$ and $f(\ldots,\xi,\ldots), q\in \SXc$. Clearly, $\SXc$ forms  a monoid with $p\cdot q$ as the product  and $\xi$ as the unit element.

\medskip

\noindent \textbf{Convention.} The frequently occurring phrases {\em deterministic top-down} and {\em nondeterministic top-down} will be abbreviated to DT and NDT, respectively.

\medskip

The DT-recognizable tree languages are characterized by the labeled paths appearing in their trees. These paths will play an important role also in our study of fuzzy DT-recognizable tree languages. The paths in a $\SX$-tree are formally defined using the {\em path alphabet}
\[
\Gamma \, := \, \bigcup \{\Sigma_m \times [m] \mid m \in \rSi\}.
\]
A pair $(f,i) \in \Gamma$ is written simply as $f_i$. In what follows, $\Gamma$ is always the path alphabet of the ranked alphabet $\Si$ considered.
 The set $\delta(t)\se \GXt$ of \emph{paths} in a $\SX$-tree $t$ is defined by
\begin{itemize}
  \item[{\rm (1)}] $\delta(x) = \{x\}$ for $x\in X$, and
  \item[{\rm (2)}] $\delta(t) = f_1\delta(t_1) \cup \ldots \cup f_m\delta(t_m)$ for $t = f(t_1,\ldots,t_m)$.
\end{itemize}
When we regard $\Gamma$ as a unary ranked alphabet, then the \emph{path language} $\delta(T) := \bigcup\{\delta(t) \mid t \in T\} (\se \GXt)$ of a $\SX$-tree language $T$ is  a set of unary trees in Polish form. The \emph{path closure} $\Delta(T) := \delta^{-1}(\delta(T))$ of $T\se \SXt$ consists of all the $\SX$-trees $t$ such that $\delta(t) \se \delta(T)$, and
$T$ is \emph{path closed} if $T=\Delta(T)$. Note that $\delta^{-1}(U) := \{t\in \SXt \mid \delta(t) \se U\}$ is path closed for every $U\se \GXt$.
Often it is convenient to treat $\Ga$ as an ordinary alphabet and view paths in $\SX$-trees as expressions  $wx$, where $w \in \Gamma^*$ and $x\in X$.

\begin{example}\label{ex:Paths} Let $\Si = \{f/2,g/1\}$ and $X = \{x,y\}$. Now $\Gamma = \{f_1,f_2,g_1\}$. If $t = f(g(f(x,x)),y)$, then $\delta(t) = \{f_1g_1f_1x, f_1g_1f_2x, f_2y\}$. For the $\SX$-tree language $T = \{f(x,y),f(y,x)\}$, we get  $\delta(T) = \{f_1x,f_2y,f_1y,f_2x\}$ and $\Delta(T) = \{f(x,y),f(y,x),f(x,x),f(y,y)\}$.
\end{example}

Recall that a \emph{$\Sigma$-algebra} $\algA$ consists of a nonempty set $A$ and a $\Sigma$-indexed family of operations $f^\calA: A^m \ra  A$, where the arity $m$ is that of the symbol $f (\in \Si_m)$.
The $\SX$-{\em term algebra} $\calT_\Si(X) = (T_{\Sigma}(X), \Sigma)$ is defined by $f^{{\mathcal T}_{\Sigma}(X)}(t_1, \ldots , t_m)  \,=\, f(t_1,\ldots , t_m)$ ($m \in \rSi, f \in \Sigma_m$, $t_1, \ldots , t_m \in T_{\Sigma}(X)$).
It is well known that the maps $\SXt \ra  \SXt, t \mapsto p(t),$ where $p$ is a $\SX$-context, are precisely the \emph{translations} of the term algebra $\calT_\Si(X)$ (cf. \cite{BuSa81,Coh81, Ste05}, for example).

\section{Top-down algebras and recognizers}\label{se:Top-down algebras and recognizers}
All the tree recognizers to be considered in this paper are essentially deterministic or nondeterministic  finite ``top-down algebras'' equipped with  some starting and acceptance mechanisms. For easier reference, we present all the basic definitions and facts concerning such algebras in this section.

A {\em \textup{(}finite\textup{)} DT $\Sigma$-algebra} $\algA$  consists of a nonempty (finite) set $A$ and a $\Sigma$-indexed family of {\em top-down operations} $f_{\mathcal{A}} :A \ra
A^m \ \ \ (f \in \Sigma),$ where the arity $m$ is that of $f (\in \Sigma_m)$.

Subalgebras, homomorphisms, congruences, quotient algebras and direct products can be defined for DT algebras in a natural way, and the usual basic relations hold between these notions \cite{GeSt78,GeSt84,Vir80,Esi86}. For example, the \emph{direct product}
$\calA\times\calB = (A\times B,\Si)$  of two DT $\Si$-algebras $\algA$ and $\algB$ is the DT $\Si$-algebra  such that for any $m\in \rSi$, $f\in \Si_m$, $a\in A$ and $b\in B$,
$
f_{\calA\times\calB}((a,b)) = ((a_1,b_1),\ldots,(a_m,b_m)),
$
where $(a_1,\ldots,a_m) = f_\calA(a)$ and $(b_1,\ldots,b_m) = f_\calB(b)$.

Finite DT $\Si$-algebras will serve as top-down tree automata, and their elements are then called \emph{states}. If the state of $\calA$ at a node $\nu$ of a tree is $a$ and the label of that node is $f\in \Si_m$, then $\calA$ enters the $m$ descendant nodes of $\nu$ in the respective states $a_1,\ldots,a_m$, where $(a_1,\ldots,a_m) = f_\calA(a)$. One way to represent the computations of DT algebras is offered by run trees.
For any DT $\Sigma$-algebra $\algA$, let $\Si\times A$ be the ranked alphabet with $\mathrm{r}(\Si\times A) = \rSi$ and $(\Si\times A)_m = \Si_m \times A$ for every $m\in \rSi$.  The \emph{run (tree)} $\run(\calA,t,a)$ of $\calA$ on a tree $t \in \SXt$ starting in a state $a\in A$ is a  $(\Sigma\times A)(X \times A)$-tree defined as follows:
\begin{enumerate}
    \item[(1)] $\run(\calA,x,a) = (x,a)$ for $x \in X$, and
    \item[(2)]  $\run(\calA,f(t_1, \dots, t_m),a) = (f,a)(\run(\calA,t_1,a_1), \dots,\run(\calA,t_m,a_m))$ if $f_{\calA}(a) = (a_1, \dots, a_m)$.
\end{enumerate}
Since the behaviors of the recognizers considered here are defined in terms of the leaf symbols appearing in run trees, we introduce the abbreviation $\lr$ for the composition of the functions $\leaf$ and $\run$, i.e.,
\[
\lr(\calA,t,a) := \leaf(\run(\calA,t,a))
\]
for any  $t\in \SXt$ and $a\in A$.

Any path word $w\in \Gamma^*$ induces a mapping $w^\calA : A \ra  A$ as follows.
\begin{itemize}
  \item[(1)] $a\ve^\calA = a$ for every $a\in A$.
  \item[(2)] $a(f_iu)^\calA = \pr_i(f_\calA(a))u^\calA$ for any $a\in A$, $f_i\in \Gamma$ and $u\in \Gamma^*$.
\end{itemize}

When $\calA$ is viewed as a top-down tree automaton, $aw^\calA$ is the state in which $\calA$ reaches the leaf at the end of the path described by $w$ if it starts  in state $a$ at the root of a tree containing that path.

A  \emph{(finite) NDT $\Si$-algebra} $\nalgA$ consists of a (finite) nonempty set $A$ and a $\Si$-indexed family of \emph{NDT operations} $f_\frA : A \ra  \wp(A^m)$ ($f\in \Si_m$).
We associate with any $w \in \Ga^*$ a mapping $w^\frA : A \ra  \wp(A)$ as follows:
\begin{itemize}
  \item[(1)] $a \ve^\frA = \{a\}$ for every $a\in A$, and
  \item[(2)] $a ( f_iu)^\frA = \bigcup \{\pr_i(\bfa) u^\frA \mid \bfa \in f_\frA(a)\}$ for any $a\in A$,  $f_i\in \Ga$ and $u \in \Ga^*$.
\end{itemize}
Now $aw^\frA$ can be interpreted as the set of states in which $\frA$ may reach the leaf at the end of the path represented by the word $w$ when started in state $a$ at the root of a tree containing the path.
The maps $w^\frA$ are extended to maps $\wp(A) \ra  \wp(A)$ in the natural way: $Hw^\frA := \bigcup\{aw^\frA \mid a \in H\}$ for any $w\in \Ga^*$ and $H \se A$. Note that for $w = f_iu$, we have $a w^\frA = \pr_i(f_\frA(a))u^\frA$.

The \emph{subset algebra} $\pA$ of an NDT $\Si$-algebra $\nalgA$ is the DT $\Si$-algebra such that for all $m\in \rSi$, $f\in \Si_m$ and $H \se A$, $f_{\wp\frA}(H) = (H_1,\ldots,H_m)$ with
$H_i := \bigcup\{\pr_i(f_\frA(a)) \mid a\in H\}  (i\in [m])$.

It is easy to prove the following lemma by induction on the word $w$.

\begin{lemma}\label{le:HwpA = Un awA}
If $\nalgA$ is any NDT $\Si$-algebra, then $Hw^{\wp\frA} = Hw^\frA$ for any $H\se A$ and $w\in \Ga^*$.
\end{lemma}

Among the numerous devices for defining the recognizable, or regular, tree languages, the NDT tree recognizers are the most convenient ones here. For general presentations of the theory of finite tree automata and regular tree languages, the reader may consult the references \cite{TATA,Eng75,GeSt84,GeSt97}.

An \emph{NDT $\SX$-recognizer}, $\recNA$ consists of a finite NDT $\Si$-algebra $\nalgA$,  a set $I\se A$ of \emph{initial states}, and a \emph{final state assignment} $\alpha : X \ra  \wp(A)$. The sets $T(\NA,a)$ of $\SX$-trees accepted by $\NA$ starting at the root in a state $a\in A$ are defined as follows:
\begin{itemize}
  \item[(1)] for any $x\in X$ and $a\in A$, $x \in T(\NA,a) $ if and only if  $a\in \alpha(x)$;
  \item[(2)] for $t = f(t_1,\ldots,t_m)$ and any $a\in A$, $t\in T(\NA,a)$ if and only if $t_1\in T(\NA,a_1),\ldots,t_m\in T(\NA,a_m)$ for some $(a_1,\ldots,a_m) \in f_\frA(a)$.
\end{itemize}
The tree language \emph{recognized} by $\NA$ is $T(\NA) := \bigcup\{T(\NA,a) \mid a \in I\}$.
A $\SX$-tree language $T$ is said to be \emph{recognizable}, or \emph{regular}, if $T = T(\NA)$ for some NDT $\SX$-recognizer $\NA$. Let $Rec = \{Rec(\Si,X)\}$ be the family of all regular tree languages.

It is well known that the family $Rec$ is also defined by both nondeterministic and deterministic bottom-up tree recognizers, but that deterministic top-down tree recognizers yield a proper subfamily of it.

A \emph{DT $\SX$-recognizer}, is a system $\recA$, where $\algA$ is a finite DT $\Si$-algebra, $a_0\in A$ is the \emph{initial state}, and $\alpha : X \ra  \wp(A)$ is the \emph{final state assignment}. For any $a\in A$, let $T(\bfA,a)$ be the set of all $\SX$-trees $t$ such that $b\in \alpha(x)$ for every $(x,b)$ in $\lr(\calA,t,a)$. The tree language \emph{recognized} by $\bfA$ is $T(\bfA) := T(\bfA,a_0)$. A $\SX$-tree language $T$ is said to be \emph{DT-recognizable} if $T = T(\bfA)$ for some DT $\SX$-recognizer $\bfA$. Let $DRec = \{DRec(\Si,X)\}$ be the family of DT-recognizable tree languages.

The tree language recognized by  $\bfA$ may be defined also in terms of paths:
$
T(\bfA) = \{t \in \SXt \mid  a_0w^\calA \in \alpha(x) \text{ for every } wx \in \delta(t)\}.
$

It is well known that $DRec$ is properly included in the family $Rec$ of all recognizable tree languages.
As shown in \cite{MaMo69},  a regular tree language is DT-recognizable if and only if it is path closed.
This implies that some very simple tree languages, like the set $\{f(x,y),f(y,x)\}$ considered in Example \ref{ex:Paths}, are not DT-recognizable. For more on
DT-recognizable tree languages cf. \cite{Cou78,GeSt84,GeSt97,Jur92,Vir80} and especially \cite{Jur95}.

\section{$\calL$-fuzzy tree languages}\label{se:L-fuzzy tree languages}

We shall consider fuzzy tree languages with membership degrees in a bounded lattice $\calL = (L,\leq)$. Thus they are \emph{$\calL$-fuzzy sets} in the sense of Goguen \cite{Gog67}. However, here we usually need neither distributivity nor completeness because in our fuzzy DT tree automata the degrees of acceptance are computed using the meet-operation $\land$ only, and the subsets of $L$ associated with these automata are always finite.  When distributivity or completeness is needed, this will be explicitly noted. In Section \ref{se:Path closure} the lattice is  a bounded chain $\calC$.
The classical fuzzy sets of Zadeh \cite{Zad65} are obtained when $\calC$ is the real interval $[0,1]$ with the usual $\le$-relation. It is well known that chains are distributive and that distributive lattices are \emph{locally finite} (cf. \cite{Gra71}, for example.)

\medskip

\noindent \textbf{Convention.} In what follows, unless stated otherwise, $\calL = (L,\leq)$ is a nontrivial bounded lattice. Its least and greatest element are denoted by 0 and 1, respectively.
In contexts involving constructions or decidability, it is assumed that all needed  meets and joins in $\calL$ can be effectively formed.

\medskip

An \emph{$\calL$-fuzzy $\SX$-tree language} is any map $\Phi : \SXt \ra  L$.
The \emph{support} of $\Phi$ is the $\SX$-tree language
$\supp(\Phi) := \{t\in \SXt \mid \Phi(t) > 0\}$.
If the support is finite, say $\supp(\Phi) = \{t_1,\ldots,t_n\}$, we may write
$
\Phi = \{t_1/\Phi(t_1),\ldots,t_n/\Phi(t_n)\}
$.
The \emph{range} of $\Phi$ is the set $\ran(\Phi) := \{\Phi(t) \mid t \in \SXt\}$.  If $\ran(\Phi) \se \{0,1\}$, then $\Phi$ is said to be \emph{crisp}. For any $c\in L$, let $\widetilde{c}$ be the constant $\calL$-fuzzy function $\SXt \ra  L, t \mapsto c$.

With any $\SX$-tree language $T$, we associate the crisp $\calL$-fuzzy $\SX$-tree language  $T^\chi$, the \emph{characteristic function} of $T$, defined by $T^\chi(t) = 1$ for $t\in T$, and $T^\chi(t) = 0$ for $t\notin T$.
The following facts are obvious.

\begin{remark}\label{re:supp and char} If $\Phi : \SXt \ra  L$ is crisp, then $\supp(\Phi)^\chi = \Phi$, and $\supp(T^\chi) = T$ for every $T \se \SXt$.
\end{remark}

All the usual operations on fuzzy sets  apply to $\calL$-fuzzy $\SX$-tree languages, too. In particular, the \emph{union} $\Phi \cup \Psi$ and the \emph{intersection} $\Phi \cap \Psi$ of two $\calL$-fuzzy $\SX$-tree languages $\Phi$ and $\Psi$ are defined by
$(\Phi \cup \Psi)(t) = \Phi(t) \vee \Psi(t)$ and $(\Phi \cap \Psi)(t) = \Phi(t) \wedge \Psi(t) \: (t\in \SXt)$.
The set $L^{\SXt}$ of all $\calL$-fuzzy $\SX$-tree languages forms with respect to the inclusion relation defined by
\[
\Phi \se \Psi \: \LLRa \: (\forall t \in \SXt)(\Phi(t) \leq \Psi(t))
\]
a lattice in which $\Phi \vee \Psi = \Phi \cup \Psi$ and $\Phi \wedge \Psi = \Phi \cap \Psi$.

A \emph{family of $\calL$-fuzzy tree languages} is a map $\calF$ that assigns to each pair $\Si, X$ a set $\calF(\Si,X)$ of $\calL$-fuzzy $\SX$-tree languages. Again, we write $\calF = \{\calF(\Si,X)\}$ and define inclusion, unions and intersections by the natural alphabetwise conditions.

It seems that top-down fuzzy tree recognizers have not received any attention in the literature while bottom-up fuzzy tree recognizers appear in various forms (cf. \cite{InFu75,EsLi07, BoBo10}, for example).
The recognizers to be defined below become essentially the $\calL$-fuzzy bottom-up $\SX$-recognizers of \'Esik and Liu \cite{EsLi07} when we exchange the initial and final states, and interpret the top-down transition relations as bottom-up transition relations.

In a \emph{general $\calL$-NDT $\SX$-recognizer} $\recNG$, $A$ is a finite nonempty set of \emph{states}, $\gamma = (\gamma_f)_{f\in \Si}$ is a family of $\calL$-fuzzy \emph{transition relations} $\gamma_f : A^{m+1} \ra  L$ \; ($m \in \rmr(\Si), f\in \Si_m$), $\iota : A \ra  L$ is an $\calL$-fuzzy set of \emph{initial states}, and $\om = (\om_x)_{x\in X}$ is a family of $\calL$-fuzzy sets of \emph{final states} $\om_x : A \ra  L$.

For each $a\in A$, we define $\Phi_{\bfN\bfG,a} : \SXt \ra  L$ as follows. For any  $x\in X$, let  $\Phi_{\bfN\bfG,a}(x) = \om_x(a)$, and for $t = f(t_1,\ldots,t_m)$, let
\begin{align*}
\Phi_{\bfN\bfG,a}(t) = \bigvee\{\gamma_f(a,a_1,\ldots,a_m) \land \Phi_{\bfN\bfG,a_1}(t_1)&\land\ldots\land\Phi_{\bfN\bfG,a_m}(t_m) \mid \\  &a_1,\ldots,a_m \in A\}.
\end{align*}

The $\calL$-fuzzy tree language \emph{recognized} by $\mathbf{NG}$ is given by
\[
\Phi_{\NG}(t) = \bigvee\{\iota(a) \land \Phi_{\bfN\bfG,a}(t) \mid a \in A\} \quad (t\in \SXt).
\]
An $\calL$-fuzzy $\SX$-tree language $\Phi$ is said to be \emph{recognizable}, or \emph{regular}, if $\Phi = \Phi_{\bfN\bfG}$ for some $\calL$-NDT $\SX$-recognizer $\bfN\bfG$. Let $Rec_\calL = \{Rec_\calL(\Si,X)\}$ be the family of recognizable $\calL$-fuzzy tree languages.

For the rest of this section $\calL$ is assumed to be  distributive and, therefore, locally finite. This implies that for any general $\calL$-NDT $\SX$-recognizer $\recNG$, the sublattice of $\calL$ generated by the set
\begin{align*}
K_{\NG} := \{\gamma_f(a,a_1,\ldots&,a_m) \mid m\in \rSi, f\in \Si_m, a,a_1,\ldots,a_m\in A\}\\
&\cup \{\om_x(a) \mid x\in X, a\in A\}
\end{align*}
of the elements of $\calL$ appearing in the definition of $\NG$ is finite, and hence $\ran(\Phi_\NG)$ is finite.
Our fuzzy NDT tree automata can then be simplified by eliminating fuzziness from the initial states and the transitions. A similar fact for fuzzy finite string automata was shown in \cite{MaSSY95} and in \cite{Bel02}   (cf. also \cite{MoMa02,LiPe05}).

We define  an \emph{$\calL$-NDT $\SX$-recognizer} as a system $\recNF$, where $\nalgA$ is a finite NDT $\Si$-algebra, $I \se A$ is the set of \emph{initial states}, and $\om = (\om_x)_{x\in X}$ is a family of  $\calL$-fuzzy sets of \emph{final states} $\om_x : A \ra  L$.
For each $a\in A$, we define $\Phi_{\NF,a} : \SXt \ra  L$ as follows:
\begin{itemize}
  \item[(1)] $\Phi_{\NF,a}(x) = \om_x(a)$ for $x\in X$, and
  \item[(2)] $\Phi_{\NF,a}(t) = \bigvee\{ \Phi_{\NF,a_1}(t_1) \land \ldots\land\Phi_{\NF,a_m}(t_m) \mid (a_1,\ldots,a_m) \in f_\frA(a)\}$ for $t = f(t_1,\ldots,t_m)$.
\end{itemize}
For any $H \se A$, let $\Phi_{\NF,H} := \bigcup\{\Phi_{\NF,a} \mid a \in H\}$.
The $\calL$-fuzzy $\SX$-tree language \emph{recognized} by $\NF$ is $\Phi_{\NF} := \Phi_{\NF,I}$.

Two $\calL$-fuzzy tree recognizers $\bfF$ and $\bfG$ (of any kind) are \emph{equivalent} if they define the same $\calL$-fuzzy tree language. This is expressed by writing $\bfF \equiv \bfG$.
It is clear that any $\calL$-NDT $\SX$-recognizer can be redefined as an equivalent general $\calL$-NDT $\SX$-recognizer. We shall now prove the converse.

\begin{proposition}\label{pr:NDT equiv to general NDT}
Let $\calL$ be distributive. For any general $\calL$-NDT $\SX$-recognizer there is an equivalent $\calL$-NDT $\SX$recognizer.
\end{proposition}

\begin{proof}
Consider any general $\calL$-NDT $\SX$-recognizer $\recNG$. Let $D$ be the finite sublattice of $\calL$ generated by $K_\NG$.
We  define an $\calL$-NDT $\SX$-recognizer $\NF = (\frB,I,\pi)$  as follows:
\begin{enumerate}
  \item $\frB = (A \times D,\Si)$ is the NDT $\Si$-algebra in which, for all $m\in \rSi$, $f\in \Si_m$, and $(a,d)\in A\times D$,
\begin{align*}
f_\frB((a,d)) = \{((a_1,d\land c),\ldots,(a_m,d\land c)) \mid\ &a_1,\ldots,a_m \in A,\\
&c = \gamma_f(a,a_1,\ldots,a_m)\}.
\end{align*}
  \item $I = \{(a,\iota(a)) \mid a \in A\}$.
  \item $\pi_x((a,d)) = \om_x(a)\land d$ for all $x\in X$ and $(a,d) \in A \times D$.
\end{enumerate}

First we prove by tree induction that $\Phi_{\NF,(a,d)}(t) = d \land \Phi_{\NG,a}(t)$ for all $t\in \SXt$ and $(a,d) \in A\times D$.

For  $x\in X$, $\Phi_{\NF,(a,d)}(x) = \pi_x((a,d)) = \om_x(a)\land d = d\land \Phi_{\NG,a}(x)$.
If $t = f(t_1,\ldots,t_m)$, then
\begin{align*}
\Phi_{\NF,(a,d)}(t) &= \bigvee\{\Phi_{\NF,(a_1,d\land c)}(t_1) \land \ldots \land \Phi_{\NF,(a_m,d\land c)}(t_m) \mid a_1,\ldots,a_m\in A\}\\
&= \bigvee\{(d\land c) \land \Phi_{\NG,a_1}(t_1)\land\ldots \land \Phi_{\NG,a_m}(t_m) \mid a_1,\ldots,a_m\in A\}\\
&= d \land \bigvee\{c\land \Phi_{\NG,a_1}(t_1) \land \ldots \land \Phi_{\NG,a_m}(t_m) \mid a_1,\ldots,a_m\in A\}\\
&= d\land \Phi_{\NG,a}(t),
\end{align*}
where $c$ always stands for the current $\gamma_f(a,a_1,\ldots,a_m)$.
We may now conclude that
\[
\Phi_{\NF}(t) = \bigvee\{\Phi_{\NF,(a,\iota(a))}(t) \mid a \in A\} = \bigvee\{\iota(a) \land \Phi_{\NG,a}(t) \mid a \in A\}
 = \Phi_{\NG}(t)
\]
for every $t\in \SXt$.
\end{proof}

It is common to identify a crisp fuzzy language with its support, and here this  is justified by Remark \ref{re:supp and char}. It is easy to see that $\Phi \in Rec_\calL(\Si,X)$ implies $\supp(\Phi) \in Rec(\Si,X)$, and that $T \in Rec(\Si,X)$ implies  $T^\chi \in Rec_\calL(\Si,X)$. Hence Remark \ref{re:supp and char} yields also the following facts.

\begin{lemma}\label{le:crisp rec and rec} Let $\calL$ be distributive. If $\Phi : \SXt \ra  L$ is crisp, then $\Phi \in Rec_\calL(\Si,X)$ if and only if $\supp(\Phi) \in Rec(\Si,X)$. On the other hand, a $\SX$-tree language $T$ is recognizable if and only if $T^\chi \in Rec_\calL(\Si,X)$.
\end{lemma}

For the following lemma it would actually suffice to assume that $\calL$ is locally finite.

\begin{lemma}\label{le:Weak Pumping for NDT recognizers} Assume that $\calL$ is distributive. Let $\NF$ be an $n$-state $\calL$-NDT $\SX$-recognizer, and let $h := |\ran(\Phi_{\NF})|$. For any $\SX$-tree $t$,  there exists a tree $u\in \SXt$ such that $\hg(u) \leq h^n$ and $\Phi_\NF(u) = \Phi_\NF(t)$.
\end{lemma}

\begin{proof}
Let us define an equivalence relation $\equiv$ on $\SXt$ by the condition
\[
    t \equiv u \;\LRa\; (\forall a \in A)  \Phi_{\NF,a}(t) = \Phi_{\NF,a}(u) \qquad (t, u \in \SXt).
\]
The number of  $\equiv$\,-\,classes is at most $h^n$, and hence  $\equiv$ is of finite index.

Consider any $\SX$-tree $t$. If $\hg(t) \leq h^n$, we may let $u := t$, so assume that $\hg(t) \ge h^n+1$. Following some maximum length path in $t$, we can find $p$, $q \in \SXc$ and $s \in \SXt$ such that $t = p \cdot q \cdot s$, $\dt(q) \ge 1$ and $q \cdot s \equiv s$. We shall show by induction on $p$ that $\Phi_{\NF,a}(p \cdot q \cdot s) = \Phi_{\NF,a}(p \cdot s)$ for any $a \in A$.

If $p = \xi$, then $\Phi_{\NF,a}(p \cdot q \cdot s) = \Phi_{\NF,a}(q \cdot s) = \Phi_{\NF,a}(s) = \Phi_{\NF,a}(p \cdot s)$ since $q \cdot s \equiv s$.

For $p = f(u_1, \dots, u_m) \cdot r$ , where $u_i = \xi$ and $r \in \SXc$,
\begin{align*}
    \Phi_{\NF,a}(p \cdot q \cdot s) &= \Phi_{\NF,a}(f(u_1, \dots, u_m) \cdot r \cdot q \cdot s) \\
    &= \bigvee\{ \Phi_{\NF,a_1}(u_1) \land \dots \land \Phi_{\NF,a_i}(r \cdot q \cdot s) \land \dots \land \Phi_{\NF,a_m}(u_m)
        \mid\\ &\hskip50mm (a_1,\dots,a_m) \in f_\frA(a)\} \\
    &= \bigvee\{ \Phi_{\NF,a_1}(u_1) \land \dots \land \Phi_{\NF,a_i}(r \cdot s) \land \dots \land \Phi_{\NF,a_m}(u_m)
         \mid\\ &\hskip50mm(a_1,\dots,a_m) \in f_\frA(a)\} \\
    &= \Phi_{\NF,a}(f(u_1, \dots, u_m) \cdot r \cdot s).
\end{align*}
Thus
$\Phi_{\NF}(t) = \bigvee\{\Phi_{\NF,a}(p \cdot q \cdot s) \mid a \in I\}
    = \bigvee\{\Phi_{\NF,a}(p \cdot s) \mid a \in I\} = \Phi_{\NF}(p \cdot s)$. By repeating this procedure finitely many times, we will get a tree $u\in \SXt$ such that $\hg(u) \leq h^n$ and $\Phi_\NF(u) = \Phi_\NF(t)$.
\end{proof}

The {\em parallel product} of two $\calL$-NDT $\Sigma X$-recognizers $\recNF$ and $\NG = (\frB,J,\pi)$  is the $(\calL \times \calL)$-NDT $\Sigma X$-recognizer  $\NH = (\frC,I \times J,\sigma)$, where  the NDT $\Si$-algebra $\frC = (A \times B,\Si)$ is such that for all $m \in \rSi$, $f \in \Si_m$, $a \in A$, $b \in B$, $f_{\frC}((a,b)) =$
\[
\{ ((a_1,b_1), \dots, (a_m,b_m)) \mid (a_1, \dots, a_m) \in f_{\frA}(a), (b_1, \dots, b_m) \in f_{\frB}(b) \},
\]
and $\si$ is defined by $\si_x((a,b)) = (\om_x(a),\pi_x(b))$ ($x \in X$, $(a,b) \in A \times B$).

\begin{lemma}\label{le:PhiNH=(PhiNF,PhiNG)} $\Phi_{\NH}(t) = \bigl(\Phi_{\NF}(t), \Phi_{\NG}(t) \bigr)$ for every $t \in \SXt$.
\end{lemma}

\begin{proof} That $\Phi_{\NH,(a,b)}(t) = (\Phi_{\NF,a}(t), \Phi_{\NG,b}(t))$ for any $t \in \SXt$, $a\in A$, and $b\in B$,
can be verified by tree induction. Hence, for any $H\se A$, $K\se B$ and $t \in \SXt$,
\begin{align*}
    \Phi_{\NH,H \times K}(t) &= \bigvee\{\Phi_{\NH,(a,b)}(t) \mid (a,b) \in H \times K \}\\
    &= \bigvee\{ \bigl(\Phi_{\NF,a}(t), \Phi_{\NG,b}(t)\bigr) \mid a \in H, b \in K \} \\
    &= \bigl( \bigvee\{ \Phi_{\NF,a}(t) \mid a \in H \}, \bigvee\{ \Phi_{\NG,b}(t) \mid b \in K \} \bigr)\\
    &= \bigl(\Phi_{\NF,H}(t), \Phi_{\NG,K}(t) \bigr).
\end{align*}
Thus,
$
    \Phi_{\NH}(t) = \Phi_{\NH,I \times J}(t) = \bigl(\Phi_{\NF,I}(t), \Phi_{\NG,J}(t) \bigr)
    = \bigl(\Phi_{\NF}(t), \Phi_{\NG}(t) \bigr)
$
for all $t \in \SXt$.
\end{proof}

\begin{proposition}\label{pr:NDTequivNDT decidable} Let $\calL$ be distributive. The \rm{Equivalence Problem} ``$\NF \equiv \NG$?'' of  $\calL$-NDT $\Si X$-recognizers is decidable.
\end{proposition}

\begin{proof} Of course, $\Phi_{\NF} \ne \Phi_{\NG}$ means that  there is a tree $t \in\SXt$ such that $\Phi_{\NF}(t) \ne \Phi_{\NG}(t)$. Let $\NH$ be the parallel product of $\NF$ and $\NG$. By the pumping lemma, there exists an integer $h$ such that for any $t$,  there is a tree $s \in \SXt$ such that $\Phi_{\NH}(t) = \Phi_{\NH}(s)$ and $\hg(s) < h$.  By Lemma~\ref{le:PhiNH=(PhiNF,PhiNG)}, $\Phi_{\NF}(t) \ne \Phi_{\NG}(t)$ if and only if $\Phi_{\NF}(s) \ne \Phi_{\NG}(s)$, and thus it suffices to check only finitely many trees $t$ for  $\Phi_{\NF}(t) \ne \Phi_{\NG}(t)$.
\end{proof}

\section{$\calL$-fuzzy DT tree recognizers}\label{se:L-fuzzy DT tree recognizers}

Let us define an \emph{$\calL$-fuzzy deterministic top-down} $\SX$\emph{-recognizer}, an \emph{$\calL$-DT $\SX$-recognizer} for short, as a system $\recF$, where $\algA$ is a finite DT $\Sigma$-algebra, $a_0\in A$ is the \emph{initial state},
and $\om = (\om_x)_{x\in X}$ is an $X$-indexed family of $\calL$-fuzzy sets of \emph{final states} $\om_x : A \ra  L$.
The $\calL$-fuzzy $\SX$-tree language $\Phi_{\bfF,a} : \SXt \ra  L$ recognized by $\bfA$ starting in a state $a\in A$ is defined as follows:
\begin{itemize}
  \item[(1)] $\Phi_{\bfF,a}(x) = \om_x(a)$ for any $x\in X$;
  \item[(2)] $\Phi_{\bfF,a}(t) = \Phi_{\bfF,a_1}(t_1)\wedge \ldots \wedge \Phi_{\bfF,a_m}(t_m)$, where $(a_1,\ldots,a_m) = f_\calA(a)$, for $t = f(t_1,\ldots,t_m)$.
\end{itemize}
The $\calL$-fuzzy tree language \emph{recognized} by $\bfF$ is  $\Phi_\bfF := \Phi_{\bfF,a_0}$.
An $\calL$-fuzzy $\SX$-tree language $\Phi$ is \emph{DT-recognizable} if $\Phi = \Phi_\bfF$ for some $\calL$-DT $\SX$-recognizer $\bfF$. Let  $DRec_\calL = \{DRec_\calL(\Si,X)\}$ be the family of DT-recognizable $\calL$-fuzzy tree languages.

Similarly as in the nondeterministic case, fuzzy state transitions would not increase the power of our $\calL$-DT tree recognizers, but since no suprema are formed, this can be proved without assuming that $\calL$ is distributive. On the other hand, we cannot allow $\calL$-fuzzy sets of initial states because $DRec_\calL$ is not closed under unions (as we shall see later).

From the definition of $\Phi_\bfF$ it is clear that for any $t\in \SXt$, the degree of acceptance $\Phi_\bfF(t)$ depends just on the states in which $\bfF$ reaches the leaves of $t$ and the labels of the leaves. Moreover, each leaf is reached by a computation of $\calA$ along some path in $t$. It is easy to prove by tree induction that
\[
\Phi_{\bfF,b}(t) = \bigwedge\{\om_x(a) \mid (x,a)\in\lr(\calA,t,b))\} = \bigwedge\{\om_x(bw^\calA) \mid wx \in \delta(t)\}
\]
for all $b\in A$ and $t\in \SXt$. For $b = a_0$ this yields the following lemma.

\begin{lemma}\label{le:Phi(t) from run tree} Let $\recF$ be an  $\calL$-DT $\SX$-recognizer. Then
\[
\Phi_\bfF(t) = \bigwedge\{\om_x(a) \mid (x,a)\in\lr(\calA,t,a_0)\} = \bigwedge\{\om_x(a_0w^\calA) \mid wx \in \delta(t)\}
\]
for every $\SX$-tree $t$.
\end{lemma}

For any $\calL$-DT $\SX$-recognizer $\recF$,  let $D_\om$ be the set of all finite meets $\om_{x_1}(a_1)\land \ldots \land \om_{x_k}(a_k)$, where $k\geq 1$, $x_1,\ldots,x_k\in X$, and $a_1,\ldots a_k \in A$, i.e., the $\land$-subsemilattice of $\calL$ generated by $R_\om := \{\om_x(a) \mid x\in X, a\in A\}$. Clearly, $R_\om$ and $D_\om$ are finite. Hence, Lemma \ref{le:Phi(t) from run tree} implies the following fact.

\begin{corollary}\label{co: range finite} For any $\calL$-DT $\SX$-recognizer $\recF$, $\ran(\Phi_\bfF) \se D_\om$, and thus the range  of any DT-recognizable $\calL$-fuzzy tree language is finite.
\end{corollary}


Let $\recF$ be an $\calL$-DT $\SX$-recognizer. For any $p \in \SXc$ and $a \in A$, let $p^\calA(a)$ be the state in which $\calA$ reaches the $\xi$-labeled leaf of $p$ if it is started in  state $a$ at the root of $p$, i.e., $(\xi,p^\calA(a))\in \lr(\calA,p,a)$, and let
\[
\Phi_{\bfF,a}(p) := \bigwedge\{\om_x(b) \mid x\in X, (x,b) \in \lr(\calA,p,a)\}.
\]

The following obvious facts can be shown by induction on the depth $\dt(p)$.

\begin{lemma}\label{le:eval of products} Let $\recF$ be an $\calL$-DT $\SX$-recognizer, and consider any $a\in A$, $p,q\in \SXc$ and $t \in \SXt$. If $p^{\calA}(a) = b$, then
\[
\Phi_{\bfF,a}(p\cdot q) = \Phi_{\bfF,a}(p) \land \Phi_{\bfF,b}(q) \text{ \: and \: } \Phi_{\bfF,a}(p\cdot t) = \Phi_{\bfF,a}(p) \land \Phi_{\bfF,b}(t).
\]
\end{lemma}

Next we present a Pumping Lemma for $\calL$-fuzzy DT tree recognizers.

\begin{lemma}\label{le:Pumping Lemma for fuzzy DT-recognizers} Let $\recF$ be an $n$-state $\calL$-DT $\SX$-recognizer and let $\ell := |D_\om|$. If $t$ is a $\SX$-tree of height $\geq (\ell + 1)n + 1$, then there exist $s\in \SXt$ and $p,q\in \SXc$ such that $t = p\cdot q \cdot s$, $\dt(q) \geq 1$ and $\Phi_\bfF(p\cdot q^k \cdot s) = \Phi_\bfF(t)$ for every $k\geq 0$.
\end{lemma}

\begin{proof} If $\hg(t)\geq (\ell + 1)n + 1$, then there is a path in $\run(\calA,t,a_0)$ on which some state $a\in A$ appears at least $\ell +2$ times. Thus we may write
\[
t = p_0\cdot q_1 \cdot \ldots \cdot q_{\ell +1} \cdot s_0
\]
for some $p_0,q_1,\ldots,q_{\ell + 1}\in \SXc$ and $s_0 \in \SXt$ such that $\dt(q_i) \geq 1$ for every $i \in [\ell +1]$ and $\bfF$ enters the roots of (the displayed occurrences of) $q_1,\ldots,q_{\ell +1}$ and $s_0$ in state $a$, i.e.,
$p_0^\calA(a_0) = q_1^\calA(a) = \ldots = q_{\ell+1}^\calA(a) = a$. By Lemma \ref{le:eval of products},
\[
\Phi_{\bfF,a_0}(p_0\cdot q_1\cdot \ldots \cdot q_i) = \Phi_{\bfF,a_0}(p_0) \land \Phi_{\bfF,a}(q_1) \land \ldots \land \Phi_{\bfF,a}(q_i)
\]
for every $i \in [\ell +1]$, and hence
\[
\Phi_{\bfF,a_0}(p_0) \geq \Phi_{\bfF,a_0}(p_0 \cdot q_1) \geq \ldots \geq \Phi_{\bfF,a_0}(p_0 \cdot q_1 \cdot \ldots \cdot q_{\ell +1}).
\]
Since these values are elements of $D_\om$, we must have
\[
\Phi_{\bfF,a_0}(p_0\cdot q_1\cdot \ldots \cdot q_{i-1}) = \Phi_{\bfF,a_0}(p_0\cdot q_1\cdot \ldots \cdot q_i)
\]
for some $i \in [\ell]$, which implies $\Phi_{\bfF,a}(q_i) \geq \Phi_{\bfF,a_0}(p_0\cdot q_1\cdot \ldots \cdot q_{i-1})$. As
\[
\Phi_\bfF(t) = \Phi_{\bfF,a_0}(p_0\cdot q_1\cdot \ldots \cdot q_{i-1}) \land \Phi_{\bfF,a}(q_i) \land \Phi_{\bfF,a}(q_{i+1}\cdot \ldots \cdot q_{\ell +1}\cdot s_0)
\]
by Lemma \ref{le:eval of products}, this means that
\[
\Phi_\bfF(t) = \Phi_\bfF((p_0\cdot q_1\cdot \ldots \cdot q_{i-1}) \cdot q_i^k \cdot (q_{i+1}\cdot \ldots \cdot q_{\ell +1}\cdot s_0))
\]
for every $k\geq 0$.
\end{proof}

The Pumping Lemma may be formulated also as follows.

\begin{lemma}\label{le:Pumpng DT-rec tree lang}
For any DT-recognizable $\calL$-fuzzy $\SX$-tree language $\Phi$, there is a constant $h$ such that if $t$ is a $\SX$-tree of height $\geq h$, then there are $p,q\in \SXc$ and $s\in \SXt$ such that $t = p\cdot q \cdot s$, $\dt(q) \geq 1$, and $\Phi(p\cdot q^k \cdot s) = \Phi(t)$ for every $k\geq 0$.
\end{lemma}

By Lemma \ref{le:Pumping Lemma for fuzzy DT-recognizers}, $\ran(\Phi_\bfF) = \{\Phi_\bfF(t) \mid t\in \SXt, \hg(t) \leq (\ell+1)n\}$.  Hence, the following result.

\begin{corollary}\label{co:ran(Phi F) can be determined} For any $\calL$-DT $\SX$-recognizer $\bfF$, the set $\ran(\Phi_\bfF)$ can be effectively determined.
\end{corollary}

Lemma \ref{le:Pumping Lemma for fuzzy DT-recognizers} and Corollary \ref{co:ran(Phi F) can be determined} yield the following decidability results. For similar results concerning tree series, cf. \cite{Bor04,FuVo09}.

\begin{proposition}\label{pr:Decidability results from Pumping Lemma}
The following questions are decidable for any $\calL$-DT $\SX$-recognizer $\bfF$.
\begin{itemize}
  \item[{\rm (a)}] The \emph{Emptiness Problem} ``$\supp(\Phi_\bfF) = \es$?''.
  \item[{\rm (b)}] The \emph{Finiteness Problem} ``Is $\supp(\Phi_\bfF)$ finite?''.
  \item[{\rm (c)}] Is $\Phi_\bfF$ constant, i.e., is $\Phi_\bfF = \widetilde{c}$ for some $c\in L$?
  \item[{\rm (d)}] Is $\Phi_\bfF$ crisp?
\end{itemize}
\end{proposition}

To treat the Equivalence Problem ``$\bfF \equiv \bfG$?'' of two  $\calL$-DT $\SX$-recognizers $\recF$ and $\recG$, we introduce the $(\calL \times \calL)$-DT $\SX$-recognizer $\bfH = (\calA \times \calB, (a_0,b_0),\sigma)$, where  $\si$ is defined by
\[
    \sigma_x((a,b)) = (\omega_x(a),\pi_x(b)))
    \qquad (x \in X, a \in A, b \in B).
\]
The following fact can easily be verified by tree induction (including induction over $\SX$-contexts, too).

\begin{lemma}\label{le:product-recognizer} $\Phi_{\bfH,(a,b)}(t) = (\Phi_{\bfF,a}(t),\Phi_{\bfG,b}(t))$ for every $t \in \SXt$ and all $a\in A$, $b\in B$.
\end{lemma}

We get the following result from Lemma~\ref{le:product-recognizer} by applying the Pumping Lemma to $\bfH$.

\begin{lemma}\label{le:reduction of tree height} One can find a constant $h$ such that for any $t\in \SXt$, there is a tree $s \in \SXt$ such that $\hg(s) < h$,
$\Phi_{\bfF}(s) = \Phi_{\bfF}(t)$ and $\Phi_{\bfG}(s) = \Phi_{\bfG}(t)$.
\end{lemma}

\begin{proposition}\label{pr:Equiv and Inclusion Decidable} The following problems are decidable for any $\calL$-DT recognizers $\bfF$ and $\bfG$.
\begin{itemize}
  \item[{\rm (a)}] The \emph{Inclusion Problem}  ``$\Phi_{\bfF} \se \Phi_{\bfG}$?''.
  \item[{\rm (b)}]  The \emph{Equivalence Problem} ``$\bfF \equiv \bfG$?''.
  \item[{\rm (c)}] The \emph{Disjointness Problem} ``$\Phi_{\bfF} \cap \Phi_{\bfG} = \widetilde{0}$?''.
\end{itemize}
\end{proposition}

\begin{proof} The inclusion $\Phi_{\bfF} \se \Phi_{\bfG}$ fails if and only if $\Phi_{\bfF}(t) \not\leq \Phi_{\bfG}(t)$ for some $t\in \SXt$. By Lemma \ref{le:reduction of tree height} this can be decided by considering a finite number of trees. The decidability of the Equivalence Problem follows directly from the decidability of the Inclusion Problem. Finally, $\Phi_{\bfF} \cap \Phi_{\bfG} = \widetilde{0}$ holds if and only if $\Phi_\bfF(t) \land \Phi_\bfG(t) = 0$ for every $t\in \SXt$, and again it suffices to consider the trees $t$ of height less than a given number.
\end{proof}

As usual, the decision methods immediately suggested by the Pumping Lemma are not necessarily computationally effective.

\section{$DRec_{\calL}$, $Rec_\calL$ and $DRec$}\label{se:DRec_L, Rec_L and DRec}

In this section we establish some connections between DT-recognizable $\calL$-fuzzy tree languages, regular $\calL$-fuzzy tree languages and the usual DT-recognizable tree languages.

\begin{lemma}\label{le:crisp drec and drec} If $\Phi : \SXt \ra  L$ is crisp, then $\Phi \in DRec_\calL(\Si,X)$ if and only if $\supp(\Phi) \in DRec(\Si,X)$. On the other hand, a $\SX$-tree language $T$ is DT-recognizable if and only if $T^\chi \in DRec_\calL(\Si,X)$.
\end{lemma}

\begin{proof}
By Remark \ref{re:supp and char} it suffices to show that $\Phi \in DRec_\calL(\Si,X)$ implies $\supp(\Phi) \in DRec(\Si,X)$, and that $T \in DRec(\Si,X)$ implies  $T^\chi \in DRec_\calL(\Si,X)$.

Let $\recF$ be an $\calL$-DT $\SX$-recognizer for which $\Phi_\bfF$ is crisp. If $\recA$ is the DT $\SX$-recognizer with $\alpha : X \ra  \wp(A)$ is defined by $\alpha(x) = \{a\in A \mid \om_x(a) = 1\}$ ($x\in X$), then obviously $T(\bfA) = \supp(\Phi_\bfF)$.

Next, let $\recA$ be a DT $\SX$-recognizer for $T$. For each $x\in X$, we define $\om_x : A \ra  L$ by setting $\om_x(a) = 1$ if $a\in \alpha(x)$, and $\om_x(a) = 0$ if $a\notin \alpha(x)$. Then $T^\chi$ is recognized by $\recF$.
\end{proof}


\begin{proposition}\label{pr:DRec_L sub Rec_L} $DRec_\calL \subset Rec_\calL$.
\end{proposition}

\begin{proof} Consider any $\calL$-DT $\SX$-recognizer $\recF$. Let $\NF = (\frB,I,\om)$ be the $\calL$-NDT $\SX$-recognizer where $\frB = (A,\Si)$ is the NDT $\Si$-algebra such that $f_\frB(a) = \{f_\calA(a)\}$ for all $f\in \Si$ and $a \in A$, and $I = \{a_0\}$. It is easy to verify by tree induction that $\Phi_{\bfN\bfF,a}(t) = \Phi_{\bfF,a}(t)$ for all $t\in \SXt$ and $a\in A$, and this implies $\Phi_\NF = \Phi_\bfF$. Thus $DRec_\calL \se Rec_\calL$.
The inclusion is proper: for example,   $\Phi = \{f(x,y)/1,f(y,x)/1\}$  is clearly recognizable  but not DT-recognizable as  $\supp(\Phi) = \{f(x,y),f(y,x)\}$  is not DT-recognizable.
\end{proof}

\begin{proposition}\label{pr:supp of FDT-language} If $0$  is $\land$-irreducible in $\calL$ (i.e., $c,d >0$ implies $c\land d > 0$), then $\supp(\Phi) \in DRec(\Si,X)$ for every $\Phi \in DRec_\calL(\Si,X)$.
\end{proposition}

\begin{proof} Let $\recF$ be any $\calL$-DT $\SX$-recognizer. If $\recA$ is the DT $\SX$-recognizer, where $\alpha(x) = \{a\in A \mid \om_x(a) > 0\}$ for each $x\in X$, then it is easy to show that for all $t\in \SXt$ and $a\in A$, $t\in T(\bfA,a)$ iff $\Phi_{\bfF,a}(t) > 0$.
Hence, $t\in T(\bfA)$ iff $\Phi_\bfF(t) > 0$, i.e., $T(\bfA) = \supp(\Phi_\bfF)$.
\end{proof}

Note that in any bounded chain $0$ is $\land$-irreducible.
The following example shows that the $\land$-irreducibility  condition is necessary unless $\Si$ is unary.

\begin{example}\label{ex:support not DT-recognizable}
Assume that there are elements $c,d\in L$ such that $c,d >0$, but $c\wedge d = 0$. Let $\Si = \{f/2\}$ and $X = \{x,y\}$. The $\calL$-fuzzy $\SX$-tree language $\Phi = \{f(x,x)/c,f(y,y)/d\}$ is recognized by the $\calL$-DT $\SX$-recognizer $\recF
$, where $A = \{a_0,a,b\}$, $f_\calA(a_0) = (a,a)$, $f_\calA(a) = f_\calA(b) = (b,b)$, $\om_x(a) = c$, $\om_y(a) = d$ and $\om_x(a_0) = \om_x(b) = \om_y(a_0) = \om_y(b) = 0$. However, $\supp(\Phi) = \{f(x,x),f(y,y)\}$ is not a DT-recognizable tree language. The example is easily adapted to any $\Si$ that contains a symbol of arity $\geq 2$.
\end{example}

For any $c\in L$, the $c$\emph{-cut} of $\Phi\in L^{\SXt}$ is the $\SX$-tree language $\Phi_{\geq c} := \{t\in\SXt \mid \Phi(t) \geq c\}$.

\begin{proposition}\label{pr:cuts are DT-recognizable} Every cut $\Phi_{\geq c}$ ($c\in L$) of a DT-recognizable $\calL$-fuzzy tree language $\Phi$ is a DT-recognizable tree language. On the other hand, if $c\in L, c >0$, then every DT-recognizable $\SX$-tree language is the $c$-cut of some DT-recognizable $\calL$-fuzzy $\SX$-tree language.
\end{proposition}

\begin{proof} Let $\recF$ be an $\calL$-DT $\SX$-recognizer for $\Phi$. If we define $\alpha $ by $\alpha(x) = \{a\in A \mid \om_x(a) \ge c\}$ ($x\in X$), then the DT $\SX$-recognizer  $\recA$ recognizes $\Phi_{\geq c}$.

Assume next that $T = T(\bfA)$ for some DT $\SX$-recognizer  $\recA$. If we define $\om = (\om_x)_{x\in X}$ by setting
$\om_x(a) = c$ for $a \in \alpha(x)$, and $\om_x(a) = 0$ for $a \notin \alpha(x)$, then $T$ is the $c$-cut of $\Phi_\bfF$ for $\recF$.
\end{proof}

The following example shows that if $|L| > 2$, then $\Phi^{-1}(d) := \{t\in \SXt \mid \Phi(t) = d\}$ is not always DT-recognizable for $\Phi  \in DRec_\calL(\Si,X)$ and $d\in L$.

\begin{example}\label{ex:Inv im of d not DT-rec} Let $d\in L$, $0 < d < 1$. It is easy to show that $\Phi = \{f(x,x)/d, f(x,y)/d, f(y,x)/d,f(y,y)/1\}$ is DT-recognizable but $\Phi^{-1}(d) = \{f(x,x),f(x,y),f(y,x)\}$ is not.
\end{example}

In \cite{Boz94} it was shown that $\Phi^{-1}(E)$ is a recognizable tree language for any recognizable tree series $\Phi$ over a finite semiring $S$ and any $E\se S$, and in \cite{FuVo09} this result is given for locally finite semirings. For the following proposition we need not assume that the lattice $\calL$ is locally finite.

\begin{theorem}\label{th:Phi^-1(d) is rec} If $\Phi \in DRec_\calL(\Si,X)$, then $\Phi^{-1}(E) \in Rec(\Si,X)$ for every $E \se L$. In particular, $\Phi^{-1}(d)\in Rec(\Si,X)$ for every $d\in L$.
\end{theorem}

\begin{proof}
Let $\recF$ be an $\calL$-DT $\SX$-recognizer for $\Phi$. Since $\Phi^{-1}(E) = \Phi^{-1}(E \cap \ran(\Phi))$, $\ran(\Phi) \se D_\om$, $D_\om$ is finite and $Rec(\Si,X)$ is closed under union, it suffices to show that $\Phi^{-1}(d)\in Rec(\Si,X)$ for every $d\in D_\om$.

Let $\NA = (\frB,I,\alpha)$ be the NDT $\SX$-recognizer, where $\nalgB$ is the NDT $\Si$-algebra with $B = A\times D_\om$ and
\[
f_\frB((a,c)) = \{((a_1,c_1),\ldots,(a_m,c_m)) \mid f_\calA(a) = (a_1,\ldots,a_m), c_1\land \ldots \land c_m = c\}
\]
for all $m \in \rmr(\Si)$, $f\in \Si_m$ and $(a,c) \in B$, $I = \{(a_0,d)\}$, and $\alpha(x)  = \{(a,\om_x(a))\mid a\in A\}$ for all $x\in X$.

To prove that $T(\NA) =  \Phi^{-1}(d)$, we show by tree induction that for all $(a,c)\in B$ and $t\in \SXt$, $\Phi_{\bfF,a}(t) = c$ if and only if $t \in T(\NA,(a,c))$.
For $x\in X$, we get
  \[
  \Phi_{\bfF,a}(x) = c \:\LRa\: \om_x(a) = c \:\LRa\: (a,c) \in \alpha(x) \:\LRa\: x\in T(\NA,(a,c)).
  \]

Let $t = f(t_1,\ldots,t_m)$, and assume first that $\Phi_{\bfF,a}(t) = c$. If $f_\calA(a) = (a_1,\ldots,a_m)$ and $\Phi_{\bfF,a_i}(t_i) = c_i$ ($i = 1,\ldots,m$), then $c = c_1\land \ldots \land c_m$, and thus $((a_1,c_1),\ldots,(a_m,c_m))\in f_\frB((a,c))$. Since $t_i \in T(\NA,(a_i,c_i))$ ($i = 1,\ldots,m$), this means that $t\in T(\NA,(a,c))$.

Conversely, if $t\in T(\NA,(a,c))$, then there exist $c_1,\ldots,c_m \in D_\om$ such that $((a_1,c_1),\ldots,(a_m,c_m))\in f_\frB((a,c))$ and $t_i \in T(\NA,(a_i,c_i))$ ($i = 1,\ldots,m$). By the definition of $f_\calB$, the former fact implies that $f_\calA(a) = (a_1,\ldots,a_m)$ and $c_1 \land \ldots \land c_m = c$. On the other hand, $\Phi_{\bfF,a_i}(t_i) = c_i$ ($i = 1,\ldots,m$) by the inductive hypothesis, and therefore
$
\Phi_{\bfF,a}(t) \:=\: \Phi_{\bfF,a_1}(t_1)\land \ldots \land \Phi_{\bfF,a_m}(t_m) \:=\: c_1\land \ldots \land c_m = c
$.
\end{proof}

We now get Corollary \ref{co:ran(Phi F) can be determined} anew: $\ran(\Phi_\bfF)$ can be found by constructing for each $\Phi_\bfF^{-1}(d)$ with $d\in D_\om$ an NDT $\SX$-recognizer $\NA$ and testing whether $T(\NA) = \es$, which can be done (cf. \cite{GeSt84}, for example).

\section{Closure properties of $DRec_\calL$}\label{se:Closure properties}

We shall now consider a number of operations on fuzzy tree languages and see which ones preserve DT-recognizability. The operations are natural extensions of classical tree language operations and also similar to the corresponding operations on tree series.  Let us begin with some positive results.

\begin{proposition}\label{pr:closed under intersection} If $\Phi, \Psi \in DRec_\calL(\Si,X)$, then  $\Phi \cap \Psi \in DRec_\calL(\Si,X)$.
\end{proposition}

\begin{proof}  Let $\recF$ and $\recG$ be any $\calL$-DT $\SX$-recognizers.  Let $\bfH = (\calA\times\calB,(a_0,b_0),\si)$ be the $\calL$-DT $\SX$-recognizer with $\si$ defined by
$\si_x((a,b)) = \om_x(a)\wedge \pi_x(b)$ ($x\in X$, $a\in A$, $b\in B$).
It is easy to verify that $\Phi_\bfH = \Phi_\bfF \cap \Phi_\bfG$. \end{proof}

By combining Proposition \ref{pr:cuts are DT-recognizable} and Lemma \ref{le:crisp drec and drec}, we get

\begin{corollary}\label{co:fuzzy cuts DT-recognzable} $\Phi_{\geq c}^{\; \chi} \in DRec_\calL(\Si,X)$ for any $\Phi \in DRec_\calL(\Si,X)$ and  $c\in L$.
\end{corollary}

For any $f\in \Si_m$, the \emph{$f$-product} (or \emph{top-concatenation}) $f(\Phi_1,\ldots,\Phi_m) : \SXt \ra  L$ of $m$ $\calL$-fuzzy $\SX$-tree languages $\Phi_1,\ldots,\Phi_m$ is defined by
\[
f(\Phi_1,\ldots,\Phi_m)(t) = \left\{
                             \begin{array}{ll}
                               \Phi_1(t_1)\wedge \ldots \wedge \Phi_m(t_m) & \hbox{\text if $t = f(t_1,\ldots,t_m)$;} \\
                               0 & \hbox{\text if $\root(t) \ne f$.}
                             \end{array}
                           \right.
\]

\begin{proposition} $f(\Phi_1,\ldots,\Phi_m) \in DRec_\calL(\Si,X)$ for all $m \in \rmr(\Si)$, $f\in \Si_m$,  $\Phi_1,\ldots,\Phi_m\in DRec_\calL(\Si,X)$.
\end{proposition}

\begin{proof} For each $i\in [m]$, let $\bfF_i = (\calA_i,a_i,\om_i)$ be an $\calL$-DT $\SX$-recognizer for $\Phi_i$.  Let  $\algA$ be the DT $\Si$-algebra, where $A$ is the disjoint union $A_1\cup \ldots \cup A_m \cup \{a_0,b\}$ and for any $g\in \Si$ and $a\in A$,
\[
    g_\calA(a) = \left\{
        \begin{array}{ll}
            (a_1,\ldots,a_m) & \text{ if } a = a_0 \text{ and } g=f;\\
            g_{\calA_i}(a) & \text{ if } a \in A_i \:(i\in [m]);\\
            (b,\ldots,b) & \text{ if } a = a_0 \text{ and } g \neq f, \text{ or } a = b.
        \end{array}
        \right.
\]
Furthermore, the family $\om = (\om_x)_{x\in X}$ of fuzzy sets $\om_x : A \ra  L$  is defined so that $\om_x(a) = \om_{ix}(a)$ for $a\in A_i$ ($i\in [m]$), and $\om_x(a) = 0$  if $a = a_0$  or $a = b$.
Clearly, $\Phi_\bfF = f(\Phi_1,\ldots,\Phi_m)$ for the $\calL$-DT $\SX$-recognizer $\recF$.
\end{proof}

For any $p\in \SXc$ and $\Phi \in L^{\SXt}$, $p^{-1}(\Phi)$ and $p(\Phi)$  are the $\calL$-fuzzy $\SX$-tree languages such that, for any $t\in \SXt$,  $p^{-1}(\Phi)(t) = \Phi(p(t))$ and
\[
p(\Phi)(t) = \left\{
               \begin{array}{ll}
                 \Phi(s) & \hbox{\text if $t = p(s), s\in \SXt$;} \\
                 0 & \hbox{\text if $t\notin p(\SXt)$.}
               \end{array}
             \right.
\]

\begin{proposition}\label{pr:closed under transl and inv transl} If $\Phi \in DRec_\calL(\Si,X)$, then $p^{-1}(\Phi), p(\Phi)\in DRec_\calL(\Si,X)$  for every $p\in \SXc$.
\end{proposition}

\begin{proof} Let $\Phi = \Phi_\bfF$ for $\recF$ and let $p\in \SXc$.
First we construct an $\calL$-DT $\SX$-recognizer for $p^{-1}(\Phi)$. Let $b := p^\calA(a_0)$ and   $d := \Phi_{\bfF,a_0}(p)$.
Then we define $\pi = (\pi_x)_{x\in X}$ by $\pi_x(a) = \om_x(a)\land d$ \: ($x\in X, a\in A$), and let $\bfG = (\calA,b,\pi)$. For any $t\in \SXt$,
\begin{align*}
       \Phi_{\bfG}(t) &= \; \bigwedge\{\pi_x(a) \mid (x,a)\in \lr(\calA,t,b)\} = \;  \bigwedge\{\om_x(a) \mid (x,a)\in \lr(\calA,t,b)\}\land d\\
       &= \; \Phi_{\bfF,b}(t) \land \Phi_{\bfF,a_0}(p) = \Phi_{\bfF,a_0}(p\cdot t) = \; \Phi_\bfF(p(t)) = \; p^{-1}(\Phi)(t).
\end{align*}

 Let us now define an $\calL$-DT $\SX$-recognizer $\recG$ for $p(\Phi)$.
 In the DT $\Si$-algebra $\algB$, $B$ is the disjoint union of $A$, $\sub(p(a_0))$ and $\{b\}$, where $p(a_0) \in T_\Si(X\cup \{a_0\})$ is obtained from $p$ by substituting $a_0$ for $\xi$. For any $f\in \Si_m$ ($m \in \rmr(\Si)$), define  $f_\calB : B \ra  B^m$ as follows:
\[
    f_\calB(a) =
        \left\{
            \begin{array}{ll}
            f_\calA(a) &\text{ if } a\in A;\\
            (r_1,\ldots,r_m) &\text{ if } a = f(r_1,\ldots,r_m) (\in \sub(p(a_0)));\\
            (b,\ldots,b) &\text{ if } a = g(r_1,\ldots,r_k) \text{ for some } g\in \Si, g\neq f;\\
            (b,\ldots,b) &\text{ if } a = b.
           \end{array}
    \right.
\]
 The initial state is $b_0 = p(a_0)$, and for each $x\in X$, $\pi_x$  is defined by setting
$\pi_x(x) := 1$ if $x\in \sub(p(a_0))$, $\pi_x(r) := 0$ for any $r\in \sub(p(a_0)), r \neq x$, and $\pi_x(a) := \om_x(a)$ for $a\in A$, and $\pi_x(b) = 0$.

The recognizer $\bfG$ operates on any input tree $t\in \SXt$ as follows. Starting at the root in state $p(a_0)$, it first checks whether $t$ is of the form $p(s)$. If not, it reaches some node in state $b$ or a leaf in a state $r\in \sub(p(a_0))$ distinct from the label of that leaf. In both cases, $\Phi_\bfG(t) = 0$. On the other hand, if $t = p(s)$ for some $s\in \SXt$, then $\bfG$ reaches the root of (the displayed occurrence of) $s$ in state $a_0$ and simulates then $\bfF$ on $s$. Since $\bfG$ assigns the value $1$ to all the $X$-labeled leaves of $p$, we get $\Phi_\bfG(t) = \Phi_\bfF(s) = p(\Phi)(t)$.
\end{proof}

A {\em tree homomorphism} $h: \SXt \to T_{\Omega}(Y)$ is defined (cf. \cite{TATA,Eng75,GeSt84,GeSt97,Tha73}) by some given mappings $h_X:X \to T_{\Omega}(Y)$ and $h_m: \Sigma_m \to T_{\Omega}(Y \cup \Xi_m)$ for each $m \in \rmr(\Si)$, where $\Xi_m = \{\xi_1,\ldots,\xi_m\}$ is a set of variables disjoint from the other alphabets considered, as follows:

1. $h(x) = h_X(x)$ for $x \in X$.

2. $h(f(t_1, \dots,t_m)) =  h_m(f)[\xi_1 \la h(t_1), \dots ,\xi_m \la h(t_m)]$ for all $m \in \rmr(\Si)$, $f \in \Sigma_m$ and $t_1$, \dots, $t_m \in \SXt$, where  $h_m(f)[\xi_1 \la h(t_1), \dots ,\xi_m \la h(t_m)]$ is  obtained from $h_m(f)$ by replacing each $\xi_i$ by $h(t_i)$ ($i = 1,\ldots,m$).

The \emph{image} $h(\Phi): T_{\Omega}(Y) \to L$ of an $\calL$-fuzzy $\SX$-tree language $\Phi$ and the \emph{inverse image} $h^{-1}(\Psi): T_{\Sigma}(X) \to L$ of an $\calL$-fuzzy $\Omega Y$-tree language $\Psi$ under $h$  are defined by
\[
    h(\Phi)(t) = \bigvee\{ \Phi(s) \mid s\in \SXt, h(s) = t\} \quad (t \in T_{\Omega}(Y)),
\]
and $h^{-1}(\Psi)(t) = \Psi(h(t)) \: (t \in \SXt)$.

For $h(\Phi)$ to be defined even for a non-complete $\calL$, it suffices to assume that $\ran(\Phi)$ is finite or that $h^{-1}(t)$ is finite for every $t\in \OYt$. The latter condition holds for any \emph{nondeleting} tree homomorphism.
Recall that $Rec$ and $DRec$ are closed under inverse tree homomorphisms \cite{GeSt84,GeSt97,Jur95}.

\begin{theorem}\label{th:inv tree homom} For any tree homomorphism $h:\SXt \to T_{\Omega}(Y)$, if $\Phi \in DRec_\calL(\Omega,Y)$, then $h^{-1}(\Phi) \in DRec_\calL(\Sigma,X)$.
\end{theorem}

\begin{proof} Let $\bfF = (\calA,a_0,\om)$ be an $\calL$-DT $\OY$-recognizer for $\Phi$ and $\calA = (A,\Omega)$ be its underlying DT $\Omega$-algebra. For all $a \in A$ and $t = f(t_1,\ldots,t_m)$ in $\SXt$, $\lr(\calA,h(t),a)$ consists of $\lr(\calA, h_m(f), a) \cap (Y \times A)$ and of the sets
\[
    \bigcup\{\, \lr(\calA,h(t_i),a') \mid (\xi_i,a') \in \lr(\calA,h_m(f),a) \,\} \quad (i \in [m]).
\]
Let $D := D_\om \cup \{1\}$ and $B := \wp(A) \times D$. The DT $\Si$-algebra $\calB = (B,\Sigma)$ is defined by setting for any $m \in \rSi$, $f \in \Sigma_m$, $H \subseteq A$, and $d \in D$,
\[
    f_{\calB}((H,d)) = \left( (H_1,d\land d_f), \dots, (H_m, d \land d_f) \right),
\]
where $H_i := \{ a' \mid (\xi_i,a') \in \lr(\calA,h_m(f),a), a\in H\}$ ($i \in [m]$), and
\[
    d_f := \bigwedge\{ \omega_y(b) \mid (y,b) \in \lr(\calA,h_m(f),a),\, a \in H,\, y \in Y \}.
\]
  \item Let $b_0 = (\{a_0\},1)$, and define $\pi = (\pi_x)_{x \in X}$ by
\[
    \pi_x((H,d)) = d \land \bigwedge\{ \Phi_{\bfF,a}(h(x)) \mid a \in H\} \qquad (x\in X, (H,d) \in B).
\]

To prove that $\Phi_{\bfG} = h^{-1}(\Phi)$ for the $\calL$-DT $\Sigma X$-recognizer $\bfG = (\calB, b_0,\pi)$, we show by tree induction that
\[
     \Phi_{\bfG,(H,d)}(t) = d \land \bigwedge\{\Phi_{\bfF,a}(h(t))\mid a \in H\}
\]
for any $(H,d) \in B$  and $t \in \SXt$.

The case $t \in X$ is obvious, so let $t = f(t_1, \dots, t_m)$. If $f_{\calB}((H,d)) = \left( (H_1,d\land d_f), \dots, (H_m, d \land d_f) \right)$, where $H_1,\ldots,H_m$ and $d_f$ are defined as above, then by the inductive assumption
\[
    \Phi_{\bfG,(H_i,d\land d_f)}(t_i) = d \land d_f \land \bigwedge \{\Phi_{\bfF,a'}(h(t_i)) \mid a' \in H_i \}
\]
for every $i \in [m]$. Hence

\begin{align*}
    &\Phi_{\bfG,(H,d)}(t) = \bigwedge_{i\in [m]} \Phi_{\bfG,(H_i,d \land d_f)}(t_i) \\
    &= d \land d_f \land \bigwedge_{i\in [m]} \bigwedge  \{\Phi_{\bfF,a'}(h(t_i)) \mid a'\in H_i\}\\
     &=d \land d_f \land \bigwedge_{i\in [m]} \bigwedge_{a \in H}\bigwedge \{\Phi_{\bfF,a'}(h(t_i)) \mid (\xi_i,a') \in \lr(\calA,h_m(f),a)\} \\
     &=d \land d_f \land \bigwedge_{a \in H}\bigwedge_{i\in [m]} \bigwedge \{\Phi_{\bfF,a'}(h(t_i)) \mid (\xi_i,a') \in \lr(\calA,h_m(f),a) \}\\
     &=d \land d_f \land
      \bigwedge_{a \in H}\bigwedge_{i\in [m]} \bigwedge \{\omega_y(b) \mid (y,b) \in \lr(\calA,h(t_i),a'),
      (\xi_i,a') \in \lr(\calA,h_m(f),a) \}\\
     &=d \land \bigwedge_{a \in H} \{ \bigwedge\{ \omega_y(b) \mid (y,b) \in \lr(\calA,h_m(f),a),\, y \in Y \} \land\\
     &\qquad \bigwedge_{i\in [m]} \bigwedge \{\omega_y(b) \mid (y,b) \in \lr(\calA,h(t_i),a'),\, (\xi_i,a') \in \lr(\calA,h_m(f),a) \} \}\\
     &=d \land \bigwedge_{a\in H} \bigwedge\{ \omega_y(b) \mid (y,b) \in \lr(\calA,h(t),a)\} = d \land \bigwedge_{a \in H} \Phi_{\bfF,a}(h(t)).
\end{align*}
For $(H,d) = (\{a_0\},1) = b_0$, we get $\Phi_\bfG(t) = \Phi_{\bfG,b_0}(t) = \Phi_{\bfF,a_0}(h(t)) = h^{-1}(\Phi)(t)$.
\end{proof}

For tree homomorphic images the closure properties of $DRec$ are quite limited (cf. \cite{Jur95}), but let us note an exception.
A tree homomorphism $h : \SXt \ra  \OYt$ is \emph{alphabetic} if for all $m \in \rmr(\Si)$ and $f\in \Si_m$, $h_m(f) = g(\xi_1,\ldots,\xi_m)$ for some $g\in \Om_m$, and $h_X(x) \in Y$ for every $x\in X$.

\begin{proposition}\label{pr:closure under inj alph tree hom} If $h : \SXt \ra  \OYt$ is an injective alphabetic tree homomorphism and  $\Phi \in DRec_\calL(\Si,X)$, then $h(\Phi) \in DRec_\calL(\Om,Y)$.
\end{proposition}

\begin{proof} Let $\recF$ be an $\calL$-DT $\SX$-recognizer for  $\Phi$. We construct an $\calL$-DT $\OY$-recognizer $\recG$ for $h(\Phi)$ as follows.
\begin{enumerate}
  \item Let $B := A \cup \{\dag\}$ ($\dag \notin A$) and $b_0 := a_0$.
  \item The operations $g_\calB : B \ra  B^m$ ($m \in \rmr(\Om)$, $g\in \Om_m$) of  $\calB = (B,\Om)$ are defined as follows. If there is an $f\in \Si_m$ for which $h_m(f) = g(\xi_1,\ldots,\xi_m)$, then $g_\calB(a) := f_\calA(a)$ for every $a\in A$, and $g_\calB(\dag) = (\dag,\ldots,\dag)$. If $g \notin h_m(\Si_m)$, then $g_\calB(b) := (\dag,\ldots,\dag)$ for every $b\in B$.
  \item Consider any $y\in Y$. If $y = h_X(x)$ for some $x \in X$, then $\pi_y(a) := \om_x(a)$ for every $a\in A$, and $\pi_y(\dag) := 0$. Otherwise, $\pi_y(b) := 0$ for every $b\in B$.
\end{enumerate}
For any $t\in \OYt$ there are two possibilities. Either $t = h(s)$ for a unique $s \in \SXt$ and $h(\Phi)(t) = \Phi(s)$, or $t \notin h(\SXt)$ and $h(\Phi)(t) = 0$. Moreover, it is clear that an $\OY$-tree $t$ is in $h(\SXt)$ if and only if every subtree of $t$ is also in the range of $h$. Thus,  $\Phi_\bfG = h(\Phi)$ follows from
\begin{itemize}
  \item[(a)] $\Phi_{\bfG,a}(h(s)) = \Phi_{\bfF,a}(s)$ for all $s\in \SXt$ and $a\in A$, and
  \item[(b)] $\Phi_\bfG(t) = 0$ for every $t \in \OYt \setminus h(\SXt)$.
\end{itemize}
Claim (a) can be verified by tree induction, and (b) follows directly from the definition of $\bfG$.
\end{proof}

The multiplication  of  $\Phi$ by a scalar $c\in L$ yields the $\calL$-fuzzy $\SX$-tree language $c.\Phi : \SXt \ra  L$  defined by $(c.\Phi)(t) = c\land \Phi(t)$.

\begin{proposition}\label{pr:scalar product} If $\Phi \in DRec_\calL(\Si,X)$, then $c.\Phi \in DRec_\calL(\Si,X)$ for every $c \in L$.
\end{proposition}

\begin{proof} If $\Phi$ is recognized by the $\calL$-DT $\SX$-recognizer $\recF$, then $c.\Phi$ is obviously recognized by $\bfG = (\calA,a_0,\pi)$, when $\pi$ is defined by $\pi_x(a) = c\land \om_x(a)$ ($x\in X, a \in A$).
\end{proof}

Finally, we note that DT-recognizability is preserved under certain changes of the valuation lattice. For any mapping $\psi : L \ra  K$ from a lattice $\calL = (L,\leq)$ to a lattice $\calK = (K,\leq)$ and any $\calL$-fuzzy $\SX$-tree language $\Phi$, let $\psi(\Phi) : \SXt \ra  K$ be the $\calK$-fuzzy $\SX$-tree language defined by $\psi(\Phi)(t) = \psi(\Phi(t))$ ($t\in \SXt$). It is known that the recognizability of tree series over semirings is preserved under semiring homomorphisms (cf. \cite{FuVo09}). Here it suffices to assume that the map $\psi$ preserves (finite) meets.

\begin{proposition}\label{pr:change of lattice} Let $\calK = (K,\leq)$ and $\calL = (L,\leq)$ be any nontrivial bounded lattices.
\begin{itemize}
  \item[{\rm (a)}] If $\calL$ is a sublattice of $\calK$, then $ DRec_\calL(\Si,X) \se DRec_\calK(\Si,X)$.
  \item[{\rm (b)}] If $\psi : L \ra  K$ is a $\land$-morphism, then $\Phi \in DRec_\calL(\Si,X)$ implies $\psi(\Phi) \in DRec_{\calK}(\Si,X)$.
\end{itemize}
\end{proposition}

\begin{proof}
Any $\calL$-DT $\SX$-recognizer may be regarded as a $\calK$-DT $\SX$-recognizer, assuming -- as we implicitly have done -- that any map $S \ra L$ can be regarded also as a map from $S$ to $K$.

Let $\Phi$ be recognized by the $\calL$-DT recognizer $\recF$. To prove (b),  we define $\pi$ by $\pi_x(a) = \psi(\om_x(a))$ ($x\in X$, $a\in A$). Then the $\calK$-DT $\SX$-recognizer $\bfG = (\calA,a_0,\pi)$ recognizes $\psi(\Phi)$. Indeed,
\begin{align*}
    \Phi_\bfG(t) &= \bigwedge\{\pi_x(a) \mid (x,a) \in \lr(\calA,t,a_0)\}
    = \bigwedge\{\psi(\om_x(a)) \mid  (x,a) \in \lr(\calA,t,a_0)\}\\
    &= \psi(\bigwedge\{\om_x(a) \mid  (x,a) \in \lr(\calA,t,a_0)\})
    = \psi(\Phi_\bfF(t)) = \psi(\Phi)(t)
\end{align*}
for every $t \in \SXt$.
\end{proof}

As one may expect, if $DRec$ is not closed under some operation, then $DRec_\calL$ is not closed under the `corresponding' $\calL$-fuzzy operation. This can be formalized as follows.

\begin{lemma}\label{le:metatheorem} Let $(T_1,\ldots,T_n) \mapsto \calO(T_1,\ldots,T_n)$ ($n\geq 1$) be an $n$-ary tree language operation, and assume that $(\Phi_1,\ldots,\Phi_n) \mapsto \wh{\calO}(\Phi_1,\ldots,\Phi_n)$ is an $n$-ary operation on $\calL$-fuzzy tree languages such that $\calO(T_1,\ldots,T_n)^\chi = \wh{\calO}(T_1^\chi,\ldots,T_n^\chi)$ for all tree languages $T_1,\ldots,T_n$ (over appropriate alphabets). If $DRec$ is not closed under $\calO$, then $DRec_\calL$ is not closed under $\wh{\calO}$.
\end{lemma}

\begin{proof} Let $T_1,\ldots,T_n \in DRec(\Si,X)$ be such that $\calO(T_1,\ldots,T_n)$ is defined but not DT-recognizable. It follows from Lemma \ref{le:crisp rec and rec} that $T_1^\chi,\ldots,T_n^\chi$ are DT-recognizable but that  $\wh{\calO}(T_1^\chi,\ldots,T_n^\chi)$ is not.
\end{proof}

To apply Lemma \ref{le:metatheorem} to an operation $\wh{\calO}$ with a known classical counterpart $\calO$, we have to verify the relation $\calO(T_1,\ldots,T_n)^\chi = \wh{\calO}(T_1^\chi,\ldots,T_n^\chi)$ and to refer  to the negative closure property of $DRec$. In all cases below, the needed results concerning $DRec$ can be found in \cite{Jur95}. In fact, the examples proving the non-closure of $DRec$ can be turned into examples for $DRec_\calL$ simply by replacing each tree language by its characteristic function.
For example, it is well known that $DRec$ is not closed under union.  Moreover, it is clear that $(T\cup U)^\chi = T^\chi \cup U^\chi$ for any $T,U \se \SXt$, and thus  $DRec_\calL$ is not closed under unions, and this can be confirmed by ``fuzzifying'' any example showing the non-closure of $DRec$. For example, the crisp $\calL$-fuzzy $\SX$-tree languages $\{f(x,y)/1\}$ and $\{f(y,x)/1\}$ are DT-recognizable while their union $\{f(x,y)/1,f(y,x)/1\}$ is not.

Next we introduce $\calL$-fuzzy forms of some known tree language operations under which $DRec$ is not closed.

For any $x \in X$, the \emph{$x$-product} $\Phi\cdot_x\Psi$ of two $\calL$-fuzzy $\SX$-tree languages $\Phi$ and $\Psi$ is defined as follows. First we define $\Phi \cdot_x s : \SXt \ra  L$ for each $s \in T_{\Sigma}(X)$ thus:
\begin{itemize}
  \item[(1)] $\Phi \cdot_x x = \Phi$ and $\Phi \cdot_x y = \{y/1\}$ for $y \in X, y \ne x$;
  \item[(2)] for $s = f(s_1, \dots, s_m)$ and any $t \in T_{\Sigma}(X)$,
\[
(\Phi \cdot_x s)(t) = \left\{
                             \begin{array}{ll}
                               (\Phi \cdot_x s_1)(t_1) \land \dots \land (\Phi \cdot_x s_m)(t_m) & \text{ if } t = f(t_1,\ldots,t_m);\\
                               0 & \text{ if } \root(t) \ne f.
                             \end{array}
                           \right.
\]
\end{itemize}

Then $\Phi\cdot_x\Psi : \SXt \ra  L$ is defined by
\[
    (\Phi \cdot_x \Psi) (t) =
    \bigvee\{ (\Phi \cdot_x s)(t) \land \Psi(s) \mid s \in T_{\Sigma}(X) \} \quad (t \in \SXt).
\]

For any $x \in X$, the \emph{$x$-iteration} of $\Phi : \SXt \ra  L$ is  the union $\Phi^{\ast x} := \bigcup\{ \Phi^{k,x} \mid k\ge 0\}$, where $\Phi^{0,x} = \{x/1\}$ and $\Phi^{k+1,x} = \Phi^{k,x} \cdot_x \Phi \cup \Phi^{k,x}$ for all $k \geq 0$.
Note that for any given $t$, the number of nonzero elements in the supremum defining $(\Phi\cdot_x\Psi)(t)$ is finite. Similarly, $\Phi^{\ast x}(t) =  \Phi^{k,x}(t)$ for some $k\geq 0$. Hence, we don't have to assume that $\calL$ is complete. It is not hard to prove the following lemma.

\begin{lemma}\label{le:chi of x-product and x-star}
$(T\cdot_x U)^\chi = T^\chi \cdot_x U^\chi$ and $(T^{\ast x})^\chi = (T^\chi)^{\ast x}$ for any $x\in X$ and $T,U \se \SXt$.
\end{lemma}

To prove the following lemma, it suffices to show that for any $t \in \SXt$, $h(T)^{\chi}(t) = 1$ if and only if $ h(T^{\chi})(t) = 1$.

\begin{lemma}\label{le:tree hom commut with chi}  $h(T)^{\chi} = h(T^{\chi})$ for any $\SX$-tree language $T$ and tree homomorphism $h:\SXt \to \OYt$.
\end{lemma}

In \cite{Jur95} it is shown (Theorem 4.2.3) that $DRec$ is not closed under unions, $x$-products, $x$-iterations,  nor under every tree homomorphism that is just alphabetic or just injective. Therefore Lemma \ref{le:metatheorem} and the above results yield the following non-closure properties of $DRec_\calL$.

\begin{proposition}\label{pr:Non-closure properties} The family $DRec_\calL$ is not closed under (a) unions, (b) $x$-products, (c) $x$-iterations,  (d) some non-alphabetic injective tree homomorphisms, and (e) some non-injective alphabetic tree homomorphisms.
\end{proposition}

In agreement with \cite{Ros71} and \cite{DiGe87}, for example, we call $\Phi : \SXt \ra  L$ an \emph{$\calL$-fuzzy subalgebra} of the term algebra $\SXta$ if $\Phi \neq \widetilde{0}$ and
$
\Phi(f(t_1,\ldots,t_m)) \geq \Phi(t_1)\land \ldots \land \Phi(t_m)
$
for all $m \in \rmr(\Si)$, $f\in \Si_m$ and $t_1,\ldots,t_m \in \SXt$.
If $\calL$ is complete, then the $\calL$-fuzzy \emph{subalgebra $[\Phi]$ generated} by an $\calL$-fuzzy $\SX$-tree language $\Phi \neq \widetilde{0}$, i.e., the least $\calL$-fuzzy subalgebra of $\SXta$ containing $\Phi$, always exists; it is the intersection of all $\calL$-fuzzy subalgebras of $\SXta$ containing $\Phi$.
A $\SX$-tree language $T$ is an ordinary subalgebra of $\SXta$ if $T \neq \es$ and $f(t_1,\ldots,t_m) \in T$ for all $m \in \rmr(\Si)$, $f\in \Si_m$ and $t_1,\ldots,t_m \in T$. Let $[T]$ denote the subalgebra of $\SXta$ generated by $T \se \SXt, T \neq \es$.

It is not hard to show that
(a) a $\SX$-tree language $T$ is a subalgebra of $\SXta$ if and only if $T^\chi$ is an $\calL$-fuzzy subalgebra of $\SXta$, and that
(b) $[T]^{\chi} = [T^{\chi}]$ for every tree language $T \se \SXt$.
In fact, (a) and (b) hold more generally for subsets of any algebra, as shown by Rosenfeld \cite{Ros71} (for groupoids).
Since $T\in DRec(\Si,X)$ does not imply $[T]\in DRec(\Si,X)$, we get the following result.

\begin{proposition}\label{pr: Non-closure under L-fuzzy subalgebra} Assume that $\calL$ is a complete lattice.
The family $DRec_\calL$ is not closed under the generation of $\calL$-fuzzy subalgebras.
\end{proposition}

\section{Fuzzy path languages and DT-recognizability}\label{se:Fuzzy path languages}

As noted above, DT-recognizable tree languages are completely defined  by the paths appearing in their trees. Here we shall discuss $\calL$-fuzzy path languages and  connect them with DT-recognizable $\calL$-fuzzy tree languages.

In what follows, $\Ga$ is again the path alphabet of our given ranked alphabet $\Si$.
By an \emph{$\calL$-fuzzy $\SX$-path language} we mean any mapping $\Lambda : \GXt \ra  L$. To introduce the fuzzy forms of the operators $\delta$ and $\delta^{-1}$, we define for any $\Phi : \SXt \ra  L$ and $\Lambda : \GXt \ra  L$ the $\calL$-fuzzy sets $\tde(\Phi) : \GXt \ra  L$ and $\tdei(\Lambda) : \SXt \ra  L$ by the respective conditions
\[
\tde(\Phi)(r) = \bigvee\{\Phi(t) \mid t\in \SXt, r \in \delta(t)\} \qquad (r \in \GXt)
\]
and
\[
\tdei(\Lambda)(t) = \bigwedge\{\Lambda(r) \mid  r \in \delta(t)\} \qquad (t \in \SXt).
\]
Restricted to crisp sets, $\tde$ and $\tdei$ match the original operators: $\supp(\tde(T^\chi)) = \delta(T)$ and $\supp(\tdei(U^\chi)) = \delta^{-1}(U)$ for any $T \se \SXt$ and $U \se \GXt$.
Note also that formally the definition of $\tde(\Phi)$ presupposes that $\calL$ is complete, but if $\Phi$ is DT-recognizable, completeness is not needed because $\ran(\Phi)$ is then finite.

\begin{lemma}\label{le:properties of tde and tdei}
The following hold for all $\Phi,\Phi' \in L^{\SXt}$ and $\Lambda,\Lambda' \in L^{\GXt}$.
\begin{itemize}
  \item[{\rm (a)}] $\Phi \se \Phi'$ implies $\tde(\Phi) \se \tde(\Phi')$, and $\Lambda \se \Lambda'$ implies $\tdei(\Lambda) \se \tdei(\Lambda')$.
  \item[{\rm (b)}] $\tde(\Phi \cup \Phi') = \tde(\Phi) \cup \tde(\Phi')$, and $\tdei(\Lambda \cap \Lambda') = \tdei(\Lambda) \cap \tdei(\Lambda')$.
  \item[{\rm (c)}] $\Phi \se \tdei(\tde(\Phi))$ and $\tde(\tdei(\Lambda)) \se \Lambda$.
  \item[{\rm (d)}] $\tde(\tdei(\tde(\Phi))) = \tde(\Phi)$ and $\tdei(\tde(\tdei(\Lambda))) = \tdei(\Lambda)$.
\end{itemize}
\end{lemma}

\begin{proof} The statements in (a) and (b) have very simple proofs. Let us prove the first part of (c). For any $t \in \SXt$,
\begin{align*}
    \tdei(\tde(\Phi))(t) &= \bigwedge\{\tde(\Phi)(r) \mid r \in \delta(t)\}\\
    &= \bigwedge\{\bigvee\{\Phi(s) \mid s\in \SXt, r \in \delta(s)\} \mid r \in \delta(t)\}.
\end{align*}
Since $\Phi(t)$ is in every set $\{\Phi(s) \mid s\in \SXt, r \in \delta(s)\}$, we may conclude that $\tdei(\tde(\Phi))(t) \geq \Phi(t)$. The second part of (c) has a similar proof.

By (a) and (c),  $\tde$ and $\tdei$ define a Galois connection between $(\SXt,\se)$ and $(\GXt,\supseteq)$, and hence (a) and (c) imply (d). For example,  (c) yields $\tde(\tdei(\tde(\Phi))) \se \tde(\Phi)$, and the converse inclusion follows from $\Phi \se \tdei(\tde(\Phi))$ by the isotonicity of $\tde$.
\end{proof}

The \emph{path closure} of $\Phi : \SXt \ra  L$ is the $\calL$-fuzzy $\SX$-tree language $\tDe(\Phi) := \tdei(\tde(\Phi))$, and $\Phi$ is said to be \emph{path closed} if $\tDe(\Phi) = \Phi$. The following corollary is an immediate consequence of Lemma \ref{le:properties of tde and tdei}.

\begin{corollary}\label{co:tDe closure operator} $\tDe$ is a closure operator on $L^{\SXt}$, i.e., for all $\Phi, \Psi \in L^{\SXt}$, \rm{(1)} $\Phi \se \tDe(\Phi)$, \rm{(2)} $\Phi \se \Psi$ implies $\tDe(\Phi) \se \tDe(\Psi)$, and \rm{(3)} $\tDe(\tDe(\Phi)) = \tDe(\Phi)$.
\end{corollary}

Let us now consider the $\calL$-fuzzy path languages of DT-recognizable $\calL$-fuzzy tree languages. First we show how  any $\calL$-DT $\SX$-recognizer defines an  $\calL$-fuzzy path language in a natural way.

As in \cite{Esi86}, we associate with any DT $\Si$-algebra $\algA$ the unary algebra $\calA^u = (A,\Gamma)$ such that $f_i^{\calA^u}(a) = \pr_i(f_\calA(a))$ for all $a\in A$ and $f_i \in \Gamma$. We may also regard $\calA^u$ as a DT $\Gamma$-algebra by treating $f_i^{\calA^u}(a)$ as a 1-tuple. In \cite{Ste17} also the converse transformation was considered:  for any $\Gamma$-algebra $\calB = (B,\Gamma)$, let $\calB^d = (B,\Si)$ be the DT $\Si$-algebra with
$
f_{\calB^d}(b) = (f_1^{\calB}(b),\ldots,f_m^{\calB}(b))
$
for all $b\in B$, $m \in \rSi$ and $f\in \Si_m$.
Since $\calA^{ud} = \calA$ for any DT $\Si$-algebra $\calA$ and $\calB^{du} = \calB$ for any $\Gamma$-algebra $\calB$, there is a bijective correspondence between DT $\Si$-algebras and $\Gamma$-algebras. As noted in \cite{Ste17}, it preserves subalgebras, homomorphisms, congruences, direct products and quotient algebras.

For any $\calL$-DT $\SX$-recognizer $\recF$, let $\bfF^u$ be the  $\calL$-DT $\GX$-recognizer $(\calA^u,a_0,\om)$, and for any $\calL$-DT $\GX$-recognizer $\recG$, let $\bfG^d$ be the $\calL$-DT $\SX$-recognizer $(\calB^d,b_0,\pi)$. Of course, $\bfF^{ud} = \bfF$ and $\bfG^{du} = \bfG$.

For any $a\in A$, $t\in \SXt$ and $wx \in \delta(t)$, $\lr(\calA^u,wx,a) = \{(x,aw^\calA)\}$, and  $(x,aw^\calA)$ is also the element of $\lr(\calA,t,a)$ appearing at the end of the path in $\run(\calA,t,a)$ described by $w$. Hence the following fact.

\begin{lemma}\label{le:lr for DT Si-algebra and hSi-algebra} Let $\algA$ be a DT $\Si$-algebra. Then
\[
lr(\calA,t,a) = \bigcup\{\lr(\calA^u,r,a) \mid r \in \delta(t)\}
\]
for all $a\in A$ and $t \in \SXt$.
\end{lemma}

An $\calL$-DT $\SX$-recognizer $\recF$ reaches the leaf at the end of the path described by a given $w \in \Gamma^*$ in state $a_0w^\calA$ independently of the tree in which the path appears.
It is therefore meaningful to regard
\[
\Lambda_\bfF : \GXt \ra  L, \: wx \mapsto \om_x(a_0w^\calA) \qquad (w\in \Gamma^*, x\in X)
\]
as the \emph{$\calL$-fuzzy path language defined by $\bfF$}.

\begin{lemma}\label{le:DT-recognizers and unary DT-recognizers}
Let $\recF$ be an $\calL$-DT $\SX$-recognizer and $\recG$ be an $\calL$-DT $\GX$-recognizer.
\begin{itemize}
  \item[\rm{(a)}] $\Lambda_\bfF = \Phi_{\bfF^u}$, and hence $\Lambda_\bfF \in DRec_\calL(\Gamma,X)$.
  \item[\rm{(b)}] $\tdei(\Phi_\bfG) = \Phi_{\bfG^d}$, and hence $\tdei(\Phi_\bfG) \in DRec_\calL(\Si,X)$.
  \item[ \rm{(c)}] $\Phi_\bfF = \tdei(\Lambda_\bfF)$, i.e., $\Phi_\bfF(t) = \bigwedge\{\Lambda_{\bfF}(r) \mid r \in \delta(t)\}$ for every $t\in \SXt$.
  \item[\rm{(d)}]  $\tde(\Phi_\bfF) \se \Lambda_\bfF$.
\end{itemize}
\end{lemma}

\begin{proof} (a) follows from the definitions of $\Lambda_\bfF$ and $\bfF^u$: for any $wx \in \GXt$,
\[
\Phi_{\bfF^u}(wx) = \bigwedge \{\om_y(a) \mid (y,a) \in \lr(\calA^u,wx,a_0)\} = \om_x(a_0w^\calA) = \Lambda_\bfF(wx).
\]
For any $r\in \GXt$, $\Phi_\bfG(r) = \pi_x(b)$  with  $\{(x,b)\} = \lr(\calB,r,b_0)$. Hence, we get (b) as follows:
\begin{displaymath}
      \begin{array}{ll}
      \Phi_{\bfG^d}(t) &= \; \bigwedge\{\pi_x(b) \mid (x,b) \in \lr(\calB^d,t,b_0)\}\\
      &= \; \bigwedge\{\pi_x(b) \mid (x,b) \in \bigcup\{\lr(\calB,r,b_0) \mid r \in \delta(t)\}\}\\
      &= \; \bigwedge\{\Phi_\bfG(r) \mid r \in \delta(t)\}\\
      &= \; \tdei(\Phi_\bfG)(t),
      \end{array}
\end{displaymath}
for every $t\in \SXt$.
Statement (c) follows from Lemma \ref{le:Phi(t) from run tree}, but now also from (a) and (b): $\tdei(\Lambda_\bfF) = \tdei(\Phi_{\bfF^u}) = \Phi_{\bfF^{ud}} = \Phi_\bfF$. Finally, for any $wx \in \GXt$,
\begin{displaymath}
      \begin{array}{ll}
      \tde(\Phi_\bfF)(wx) &= \; \bigvee\{\Phi_\bfF(t) \mid t\in \SXt, wx \in \delta(t)\}\\
      &= \;\bigvee\{\bigwedge\{\om_y(a_0v^\calA) \mid vy \in \delta(t)\} \mid t \in \SXt, wx \in \delta(t)\}\\
      &\leq \; \om_x(a_0w^\calA) = \Lambda_\bfF(wx),
      \end{array}
\end{displaymath}
and hence also (d) holds.
\end{proof}

\begin{corollary}\label{co:Fu equiv Gu implies F equiv G} If $\bfF$ and $\bfG$ are $\calL$-DT $\SX$-recognizers, then $\bfF^u \equiv \bfG^u$ implies $\bfF \equiv \bfG$. Similarly, if $\bfF$ and $\bfG$ are $\calL$-DT $\GX$-recognizers, then $\bfF \equiv \bfG$ implies $\bfF^d \equiv \bfG^d$.
\end{corollary}

\begin{proof}
The first implication follows by statements (a) and (c) of Lemma \ref{le:DT-recognizers and unary DT-recognizers}: $\Phi_\bfF = \tdei(\Lambda_\bfF) = \tdei(\Phi_{\bfF^u}) = \tdei(\Phi_{\bfG^u}) = \tdei(\Lambda_\bfG ) = \Phi_\bfG$. The second implication follows similarly from Lemma \ref{le:DT-recognizers and unary DT-recognizers} (b).
\end{proof}

The following example shows that the converses of the inclusion (d) of Lemma \ref{le:DT-recognizers and unary DT-recognizers} and the implications of Corollary \ref{co:Fu equiv Gu implies F equiv G} are not generally valid.

\begin{example}\label{ex:Lambda_bfF se tde(Phi_bfF) does not hold}
Let $\Si = \{f/2\}$, $X = \{x\}$, and let $\algA$ be the DT $\Si$-algebra such that  $A = \{a_0,a,b\}$ and $f_\calA(a_0) = (a,b)$, $f_\calA(a) = f_\calA(b) = (b,b)$.
If $\recF$ is the $\calL$-DT $\SX$-recognizer, where $\om_x(a) = 1$ and $\om_x(a_0) = \om_x(b) = 0$, then $\tde(\Phi_\bfF) = \tilde{0}$ but $\Lambda_\bfF = \{f_1x/1\}$, and hence $\Lambda_\bfF \se \tde(\Phi_\bfF)$ does not hold.

Let $\bfG = (\calA,a_0,\pi)$ be the $\calL$-DT $\SX$-recognizer with $\pi_x(b) = 1$ and $\pi_x(a_0) = \pi_x(a) = 0$. Then $\Phi_\bfF = \tilde{0} = \Phi_\bfG$, but $\Phi_{\bfF^u} = \{f_1x/1\} \neq \{f_2x/1\} = \Phi_{\bfG^u}$. Hence, $\bfF^u \equiv \bfG^u$ does not follow from $\bfF \equiv \bfG$.

That $\bfF \equiv \bfG$ does not follow from $\bfF^d \equiv \bfG^d$, can be seen by considering the $\calL$-DT $\GX$-recognizers $\bfF_1 := \bfF^u$ and $\bfG_1 := \bfG^u$, where $\bfF$ and $\bfG$ are as above. Now $\bfF_1^d = \bfF \equiv \bfG = \bfG_1^d$, but not $\bfF_1 \equiv \bfG_1$.
\end{example}

We may now characterize the $\calL$-DT tree languages in terms of $\calL$-fuzzy path languages.

\begin{theorem}\label{th:Fuzzy DT tree lang and fuzzy path lang}
An $\calL$-fuzzy $\SX$-tree language $\Phi$ is DT-recognizable if and only if $\Phi = \tdei(\Lambda)$ for some DT-recognizable $\calL$-fuzzy $\SX$-path language $\Lambda$.
\end{theorem}

\begin{proof}
If $\Phi$ is recognized by an $\calL$-fuzzy DT $\SX$-recognizer $\bfF$, then it follows from Lemma \ref{le:DT-recognizers and unary DT-recognizers} that $\Phi = \Phi_\bfF = \tdei(\Lambda_\bfF)$ with $\Lambda_\bfF \in DRec_\calL(\Gamma,X)$.
Assume then that $\Phi = \tdei(\Lambda)$, where $\Lambda = \Phi_\bfG$ for an $\calL$-DT $\GX$-recognizer $\bfG$. Now Lemma \ref{le:DT-recognizers and unary DT-recognizers} yields $\Phi = \tdei(\Phi_\bfG) = \Phi_{\bfG^d} \in DRec_\calL(\Si,X)$.
\end{proof}

\begin{corollary}\label{co:DT-rec is path closed} Any DT-recognizable tree language is path closed.
\end{corollary}

\begin{proof}  By Theorem \ref{th:Fuzzy DT tree lang and fuzzy path lang}, if $\Phi \in DRec_\calL(\Si,X)$, then  $\Phi = \tdei(\Lambda)$ for some $\Lambda \in DRec_\calL(\Ga,X)$, and hence $\tDe(\Phi) = \tdei(\tde(\tdei(\Lambda))) = \tdei(\Lambda) = \Phi$ by Lemma \ref{le:properties of tde and tdei}.
\end{proof}

\section{Path closure and DT-recognizability}\label{se:Path closure}

An ordinary regular tree language is  DT-recognizable if and only if it is path closed, and the path closure of the tree language recognized by an NDT tree recognizer is recognized by a DT tree recognizer obtained by a subset construction (cf. \cite{Jur95} or \cite{GeSt84}, for example).
We shall now prove some similar results for fuzzy tree languages.

Throughout this section we assume that the lattice of membership degrees  is a nontrivial bounded chain $\calC = (C,\leq)$.
Moreover,  $\recNF$ is always a $\calC$-NDT $\SX$-recognizer with the underlying NDT $\Si$-algebra $\nalgA$.

Since $\calC$ is a chain, the set $R_\om = \{\om_x(a) \mid x\in X, a\in A\}$ is a sublattice of $\calC$ and $\ran(\Phi_\NF) \se R_\om$.
For each $a\in A$, we define $\Lambda_{\NF,a} : \GXt \ra C$ by
\[
\Lambda_{\NF,a}(wx) = \max \{ \om_x(b) \mid  b \in a w^\frA\} \quad (w\in \Ga^*, x\in X).
\]
The $\calC$-fuzzy path language  $\Lambda_\NF : \GXt \ra C$ defined by $\NF$ is given by
\[
\Lambda_\NF(r) = \max \{\Lambda_{\NF,a}(r) \mid a \in I\} \quad (r\in \GXt).
\]

\begin{lemma}\label{le:Phi_NF(t) leq Lambda_NF(r)} Let $\NF$ be any $\calC$-NDT $\SX$-recognizer.
\begin{itemize}
  \item[{\rm (a)}] If $t\in \SXt$ and $r \in \delta(t)$, then $\Phi_\NF(t) \leq \Lambda_\NF(r)$.
  \item[{\rm (b)}] $\tde(\Phi_\NF) \se \Lambda_\NF$.
\end{itemize}
\end{lemma}

\begin{proof}
Let us first show by tree induction that $\Phi_{\NF,a}(t) \leq \Lambda_{\NF,a}(r)$ for any $t\in \SXt$, $r\in \delta(t)$ and $a\in A$.

If $t = x \in X$, then $r$ must be $x$, and hence $\Phi_{\NF,a}(t) = \om_x(a) = \Lambda_{\NF,a}(r)$.

If $t = f(t_1,\ldots,t_m)$, then $r = f_iux$ for some $i\in[m]$, $u \in \Ga^*$ and $x\in X$. Because $\calC$ is a chain, $\Phi_{\NF,a}(t) = \Phi_{\NF,a_1}(t_1) \land \ldots \land \Phi_{\NF,a_m}(t_m)$
for some $(a_1,\ldots,a_m) \in f_\frA(a)$. Moreover, $ux \in \delta(t_i)$, and therefore
$
\Lambda_{\NF,a}(r) \geq \Lambda_{\NF,a_i}(ux) \geq \Phi_{\NF,a_i}(t_i) \geq \Phi_{\NF,a}(t).
$
Now we get (a) as follows:
\[
\Phi_\NF(t) = \max\{\Phi_{\NF,a}(t) \mid a \in I\} \leq \max\{\Lambda_{\NF,a}(r) \mid a \in I\} = \Lambda_\NF(r).
\]
Statement (b) follows from (a) because
$
\tde(\Phi_\NF)(r) = \max\{\Phi_\NF(t) \mid t\in \SXt, r\in \delta(t)\} \leq \Lambda_\NF(r)
$
for every $r\in \GXt$.
\end{proof}

The \emph{subset recognizer} of a $\calC$-NDT $\SX$-recognizer $\recNF$ is the $\calC$-DT $\SX$-recognizer $\recpNF$, where
$\pi = (\pi_x)_{x\in X}$ is defined by $\pi_x(H) = \max \{\om_x(a) \mid a \in H\}$ ($x \in X, H \se A$).

\begin{proposition}\label{pr:LambdaNF=LambdapNF} $\Lambda_{\NF} = \Lambda_{\wp\NF}$ for any $\calC$-NDT $\SX$-recognizer $\NF$.
\end{proposition}

\begin{proof} Let $\recNF$. For any $wx \in \GXt$,
\begin{align*}
\Lambda_{\NF}(wx) &= \max \{\Lambda_{\NF,a}(wx) \mid  a\in I\}\\
    &= \max \{\omega_x(b) \mid b \in aw^{\frA}, a\in I\}\\
    &= \max \{\omega_x(b) \mid b \in Iw^{\frA}\}\\
    &=  \max\{\omega_x(b) \mid b \in Iw^{\wp\frA}\}\\
    &= \pi_x(Iw^{\wp\frA}) = \Lambda_{\wp\NF}(wx),
\end{align*}
where the fourth equality is justified by Lemma \ref{le:HwpA = Un awA}.
\end{proof}

The minimization theory of usual DT tree recognizers \cite{GeSt78,GeSt84} uses ``normalized'' DT tree recognizers to deal with so-called 0-states. In \cite{GeSt84} the characterization of the DT-recognizable tree languages as the path closed regular tree languages employs normalized NDT tree recognizers. Here the division of states into 0-states and states from which some tree can be accepted does not suffice since there may be several degrees of acceptance.

For any $a\in A$, let
$
M(a) := \max\{\Phi_{\NF,a}(t) \mid t \in \SXt\}.
$
Since $\Phi_{\NF,a}(t) \in R_\om$ for every $t$ and $R_\om$ is finite, $M(a)$ is well-defined. Moreover, the values $M(a)$ can be determined as follows.

For any $a\in A$, let $M_0(a) := \max\{\om_x(a) \mid x\in X\}$, and for each $k\geq 0$, let $M_{k+1}(a) :=$
\[
    M_k(a) \lor \max\{M_k(a_1) \land \ldots \land M_k(a_m) \mid f \in \Si, (a_1,\ldots,a_m)\in f_\frA(a)\}.
\]
Obviously, $M_k(a)$ is the maximal degree of acceptance of any $\SX$-tree of height $\leq k$ ($k \geq 0$) when $\NF$ starts at the root in state $a$. Since
$
M_0(a) \leq M_1(a) \leq M_2(a) \leq \ldots \leq \max\{c \mid c \in R_\om\}
$,
there is a $k$ such that $M_{k+1}(a) = M_k(a)$ for every $a\in A$. Clearly, $M(a) = M_k(a)$ for such a $k$ and any $a\in A$.

We say that $\NF$ is \emph{normalized} if $M(a_1) = \ldots = M(a_m)$ for all $m \in \rSi$, $f\in \Si_m$, $a\in A$, and $(a_1,\ldots,a_m) \in f_\frA(a)$.

\begin{theorem}\label{th:Any NDT rec equiv to normalized}
Any $\calC$-NDT $\SX$-recognizer is equivalent to a normalized
$\calC$-NDT $\SX$-recognizer.
\end{theorem}

\begin{proof} Let $\recNF$ be any $\calC$-NDT $\SX$-recognizer. We construct a normalized $\calC$-NDT
$\SX$-recognizer $\NG = (\frB,J,\pi)$ as follows.
Let $\frB = (B,\Si)$ be the NDT $\Si$-algebra, where $B :=
A\times R_\om$ and for any $m\in \rSi$, $f\in \Si_m$ and $(a,d)
\in B$,
\[
    f_\frB((a,d)) = \{\, ((a_1,d\land c),\dots,(a_m,d\land c)) \mid  (a_1,\dots,a_m) \in f_\frA(a)\,\},
\]
where $c = \min(M(a_1),\dots,M(a_m))$ for each $(a_1,\dots,a_m) \in f_\frA(a)$. Note that $c\in R_\om$ and that $R_\om$ is closed under meets.
The set of initial states is $J := \{ (a,M(a)) \mid a \in I\}$,
and let $\pi_x((a,d)) := \om_x(a)\land d$ for all $x\in X$ and
$(a,d) \in B$.
To prove that $\NG \equiv \NF$, we show by tree induction that
$\Phi_{\NG,(a,d)}(t) = \Phi_{\NF,a}(t)\land d$ for all $t\in
\SXt$ and $(a,d) \in B$.

Firstly,  $\Phi_{\NG,(a,d)}(x) = \pi_x((a,d)) =
\om_x(a)\land d = \Phi_{\NF,a}(x)\land d$ for $x\in X$.

Let $f = f(t_1,\dots,t_m)$. For each $(a_1,\dots,a_m) \in f_\frA(a)$, we let $c$ denote $\min(M(a_1),\dots,M(a_m))$. As
$\Phi_{\NF,a_1}(t_1)\land \dots \land \Phi_{\NF,a_m}(t_m) \le
c$, we get
\begin{align*}
      &\Phi_{\NG,(a,d)}(t) \\
      &= \max\{\min\{\Phi_{\NG,(a_i,d\land c)}(t_i) \mid i\in [m]\} \mid (a_1, \dots, a_m) \in f_{\frA}(a) \}\\
      &= \max\{\min\{\Phi_{\NF,a_i}(t_i)\land (d\land c) \mid i \in [m] \} \mid (a_1, \dots, a_m) \in f_{\frA}(a)\}\\
      &= \max\{ \Phi_{\NF,a_1}(t_1)\land \dots \land \Phi_{\NF,a_m}(t_m) \mid (a_1, \dots, a_m) \in f_{\frA}(a)\} \land d \\
      &=  \Phi_{\NF,a}(t) \land d.
\end{align*}
In particular, for every $t \in \SXt$,
\begin{align*}
\Phi_\NG(t) &= \max \{ \Phi_{\NG,b}(t) \mid b \in J \} = \max \{ \Phi_{\NG,(a,M(a))}(t) \mid a \in I \}\\
        &= \max \{ \Phi_{\NF,a}(t) \land M(a) \mid a \in I \} = \max \{ \Phi_{\NF,a}(t) \mid a \in I \}\\
        &= \Phi_\NF(t).
\end{align*}
To show that $\NG$ is normalized, let $((a_1,d\land
c),\dots,(a_m,d\land c)) \in f_\frB((a,d))$ be as in the
definition of $\frB$, and consider any $i \in [m]$. For any
$t\in \SXt$,
\[
\Phi_{\NG,(a_i,d\land c)}(t) = \Phi_{\NF,a_i}(t) \land d \land
c \leq M(a_i) \land d \land c = d\land c.
\]
On the other hand, if $t\in \SXt$ is a tree such that
$\Phi_{\NF,a_i}(t) = M(a_i)$, then $\Phi_{\NG,(a_i,d\land
c)}(t) = M(a_i) \land d \land c = d\land c$, and hence
$M((a_i,d\land c)) = d \land c$ for every $i \in [m]$.
\end{proof}

\begin{lemma}\label{le:Normalized NDT Phi(t) = Lambda(r)}
Let $\NF$ be a normalized $\calC$-NDT $\SX$-recognizer. For any $r\in \GXt$, there is a tree $t\in \SXt$ such that $r\in \delta(t)$ and $\Phi_{\NF}(t) = \Lambda_{\NF}(r)$.
\end{lemma}

\begin{proof} First we show by induction on $w$ that for any $wx\in \GXt$ and $a\in A$, there is a tree $t\in \SXt$ such that $wx\in \delta(t)$ and $\Phi_{\NF,a}(t) = \Lambda_{\NF,a}(wx)$.

If $w = \ve$, we may choose $t = x$.
Let then $w = f_iux$ for some $f_i\in \Ga$ and $u \in \Ga^*$. For any $a\in A$,
\[
\Lambda_{\NF,a}(wx) = \max\{\om_x(b) \mid b \in aw^\frA\} = \max\{\om_x(b) \mid b \in \pr_i(\bfa) u^\frA, \bfa \in f_\frA(a)\}.
\]
Let $\bfa = (a_1,\ldots,a_m)$ be an element of $f_\frA(a)$ for which the maximal value $\om_x(b)$ is obtained for some $b\in a_iu^\frA$. Then $\Lambda_{\NF,a}(wx) = \Lambda_{\NF,a_i}(ux)$. By the inductive assumption, there is a tree $t_i \in \SXt$ such that $ux\in \delta(t_i)$ and $\Phi_{\NF,a_i}(t_i) = \Lambda_{\NF,a_i}(ux)$. Since $\NF$ is normalized, there exists for each $j\in [m], j \neq i$, a tree $t_j \in \SXt$ such that $\Phi_{\NF,a_j}(t_j) = M(a_i) \geq \Phi_{\NF,a_i}(t_i)$. For $t := f(t_1,\ldots,t_m)$,
\begin{align*}
\Phi_{\NF,a}(t) &\geq \Phi_{\NF,a_1}(t_1) \land \ldots \land \Phi_{\NF,a_m}(t_m)\\
&= \Phi_{\NF,a_i}(t_i) = \Lambda_{\NF,a_i}(ux) = \Lambda_{\NF,a}(wx).
\end{align*}
The converse $\Phi_{\NF,a}(t) \leq \Lambda_{\NF,a}(wx)$ holds by Lemma \ref{le:Phi_NF(t) leq Lambda_NF(r)} because $wx \in \delta(t)$.

To prove the lemma itself, consider any $r\in \GXt$. By definition, $\Lambda_{\NF}(r) = \max\{\Lambda_{\NF,a}(r) \mid a \in I\}$. Let $b\in I$ be a state for which $\Lambda_{\NF}(r) = \Lambda_{\NF,b}(r)$. By the first part of the lemma, there is a $t \in \SXt$ such that $r\in \delta(t)$ and $\Phi_{\NF,b}(t) = \Lambda_{\NF,b}(r)$. For any $a\in I$, $\Phi_{\NF,a}(t) \leq \Lambda_{\NF,a}(r)$ by Lemma \ref{le:Phi_NF(t) leq Lambda_NF(r)}, and therefore
$
\Phi_{\NF,a}(t) \leq \Lambda_{\NF,b}(r) = \Phi_{\NF,b}(t).
$
This implies that $\Phi_{\NF}(t) = \Phi_{\NF,b}(t) = \Lambda_{\NF}(r)$.
\end{proof}

\begin{theorem}\label{th:deltaNF=pNF}  $\Phi_{\pNF} = \tDe(\Phi_{\NF})$ for any normalized $\calC$-NDT $\SX$-recognizer $\NF$.
\end{theorem}

\begin{proof} Let $t \in \SXt$. For every $r \in \delta(t)$, $\Phi_\NF(t) \leq \Lambda_\NF(r)$ by Lemma~\ref{le:Phi_NF(t) leq Lambda_NF(r)}, but also $\Phi_\NF(s) \leq \Lambda_\NF(r)$ for any  $s \in \SXt$ such that $r \in \delta(s)$. On the other hand, by Lemma~\ref{le:Normalized NDT Phi(t) = Lambda(r)} there is an $s \in \SXt$  such that $r \in \delta(s)$ and $\Phi_{\NF}(s) = \Lambda_{\NF}(r)$. Thus
\begin{align*}
    \tDe(\Phi_{\NF})(t) &= \min \{\max\{ \Phi_{\NF}(s) \mid r \in \delta(s)\} \mid r \in \delta(t)\}\\
    &= \min\{\Lambda_{\NF}(r) \mid r \in \delta(t)\}.
\end{align*}
By Proposition~\ref{pr:LambdaNF=LambdapNF}, $\Lambda_{\NF}(r) = \Lambda_{\pNF}(r)$ for any $r\in \GXt$. Hence  we get
\[
   \tDe(\Phi_{\NF})(t) = \min \{\Lambda_{\pNF}(r)  \mid r \in \delta(t)\} = \Phi_{\pNF}(t),
\]
where the second equality follows from Lemma~\ref{le:DT-recognizers and unary DT-recognizers}(c).
\end{proof}

\begin{theorem}\label{th:Path closure of Rec}
The path closure of any regular $\calC$-fuzzy tree language is DT-recognizable.
A regular $\calC$-fuzzy tree language is DT-recognizable if and only if it is path closed.
\end{theorem}

\begin{proof}
The first statement follows directly from Theorems \ref{th:deltaNF=pNF} and \ref{th:Any NDT rec equiv to normalized}.

Consider any $\calC$-NDT $\SX$-recognizer $\NF$. By Theorem \ref{th:Any NDT rec equiv to normalized} we may assume that $\NF$ is normalized.
If $\Phi_{\NF}$ is path closed, then $\Phi_\NF = \tDe(\Phi_\NF) = \Phi_{\wp\NF} \in DRec_\calC(\Si,X)$ by Theorem \ref{th:deltaNF=pNF}.
On the other hand, if $\Phi_\NF$ is DT-recognizable, then it is path closed by Corollary \ref{co:DT-rec is path closed}.
\end{proof}

As a further application of Theorem \ref{th:deltaNF=pNF} we get the following result.

\begin{proposition}\label{pr:DT-recognizability decidable}
It is decidable whether a given regular $\calC$-fuzzy tree language is DT-recognizable.
\end{proposition}

\begin{proof}
Consider any $\Phi \in Rec_\calC(\Si,X)$. By Theorem \ref{th:Any NDT rec equiv to normalized} we may assume that  $\Phi = \Phi_\NF$ for a given normalized $\calC$-NDT recognizer $\NF$. By Theorem \ref{th:Path closure of Rec}, $\Phi$ is DT-recognizable if and only if $\tDe(\Phi) = \Phi$. Since $\tDe(\Phi) = \Phi_{\wp\NF}$ by Theorem \ref{th:deltaNF=pNF}, this holds if and only if $\wp\NF \equiv \NF$, which is decidable by Proposition \ref{pr:NDTequivNDT decidable}.
\end{proof}

We shall now consider normalized deterministic $\calC$-fuzzy tree recognizers. Let $\recF$ be any $\calC$-DT $\SX$-recognizer. Similarly as in the nondeterministic case, we set $M(a) := \max\{\Phi_{\bfF,a}(t) \mid t\in \SXt\}$ for any $a\in A$, and
say that $\bfF$ is \emph{normalized} if for all $m \in \rSi$, $f\in \Si_m$ and $a\in A$, $f_\calA(a) = (a_1,\ldots,a_m)$ implies $M(a_1) = \ldots = M(a_m)$.

The statements of the following proposition are obtained as special cases from Theorem \ref{th:Any NDT rec equiv to normalized} and \ref{le:Normalized NDT Phi(t) = Lambda(r)}. To get (a), we have to verify that applied to a DT recognizer, the construction used in the proof of Theorem \ref{th:Any NDT rec equiv to normalized} yields a deterministic recognizer.

\begin{proposition}\label{pr:normalized L-DT recognizers} Let $\recF$ be any $\calC$-DT $\SX$-recognizer.
\begin{itemize}
  \item[{\rm (a)}] $\bfF$ is equivalent to a normalized
$\calC$-DT $\SX$-recognizer.
  \item[{\rm (b)}] If $\bfF$ is normalized, then  there is for any $r\in \GXt$, a $\SX$-tree $t$ such that $r\in \delta(t)$ and $\Phi_{\bfF}(t) = \Lambda_{\bfF}(r)$.
\end{itemize}
\end{proposition}

We may now prove for the special case at hand, the following stronger form of Lemma \ref{le:DT-recognizers and unary DT-recognizers} (d).

\begin{proposition}\label{pr:tde(Phi_F) = Lambda_F for norm F}
$\tde(\Phi_\bfF) = \Lambda_\bfF$ for any normalized $\calC$-DT $\SX$-recognizer $\bfF$.
\end{proposition}

\begin{proof}
Consider any $r \in \GXt$. By Proposition \ref{pr:normalized L-DT recognizers} (b), $\Lambda_\bfF(r) = \Phi_\bfF(t)$ for some $t \in \SXt$ such that $r \in \delta(t)$, which implies that
$
\tde(\Phi_\bfF)(r) = \max \{\Phi_\bfF(t) \mid t \in \SXt, r\in \delta(t)\} \geq \Lambda_\bfF(r)
$, and thus $\tde(\Phi_\bfF) \supseteq \Lambda_\bfF$. The converse inclusion $\tde(\Phi_\bfF) \se \Lambda_\bfF$ is given by Lemma \ref{le:DT-recognizers and unary DT-recognizers}(d).
\end{proof}

\begin{proposition}\label{pr:F equiv G iff Fu equiv Gu} If $\bfF$ and $\bfG$ are normalized $\calC$-DT $\SX$-recognizers, then $\bfF \equiv \bfG$ if and only if $\bfF^u \equiv \bfG^u$.
\end{proposition}

\begin{proof}
If $\bfF \equiv \bfG$, then
$\Phi_{\bfF^u} = \Lambda_\bfF = \tde(\Phi_\bfF) = \tde(\Phi_\bfG) = \Lambda_\bfG = \Phi_{\bfG^u}$
by Lemma \ref{le:DT-recognizers and unary DT-recognizers} (a) and Proposition \ref{pr:tde(Phi_F) = Lambda_F for norm F}. The converse implication is given by Corollary \ref{co:Fu equiv Gu implies F equiv G}.
\end{proof}

Note that Example \ref{ex:Lambda_bfF se tde(Phi_bfF) does not hold} is valid also under the assumption that $\calC$ is a chain. Hence it shows that Propositions \ref{pr:tde(Phi_F) = Lambda_F for norm F} and \ref{pr:F equiv G iff Fu equiv Gu}  do not hold in full for $\calC$-DT tree recognizers that are not normalized.
Moreover, statement (b) of Proposition \ref{pr:normalized L-DT recognizers} does not hold for the non-normalized $\calC$-DT $\SX$-recognizer $\bfF$ appearing in that example. Indeed, $\Lambda_\bfF(f_1x) = 1$, but there is no tree $t$ such that $f_1x \in \delta(t)$ and $\Phi_\bfF(t) = 1$.

The above results suggest the possibility of dealing with questions concerning $\calC$-DT tree recognizers by a reduction to unary
$\calC$-DT tree recognizers. Furthermore, if  $\Si$ is unary, then any $\calC$-DT $\SX$-recognizer $\recF$ can be turned into a finite automaton that recognizes a family $\big(L_{x,c} \se \Si^* \mid x\in X, c\in \ran(\Phi_\bfF)\big)$ of regular languages that completely determines $\Phi_\bfF$; the language $L_{x,c}$ is recognized  when the set of final states is $\{a\in A \mid \om_x(a) = c\}$. Thus the standard results and methods of the theory of finite automata (cf. \cite{Bra84,Eil74,Sak09}, for example) become applicable. For example, by Propositions \ref{pr:F equiv G iff Fu equiv Gu} and \ref{pr:normalized L-DT recognizers}, the decidability of the equivalence problem `$\bfF \equiv \bfG$?'' of $\calC$-DT tree automata could be inferred this way.

\end{document}